\documentclass[12pt,reqno]{amsart}
\usepackage{amsthm}
\usepackage{amsmath}
\usepackage{latexsym}
\usepackage{amsfonts}
\usepackage{amssymb}
\usepackage{color}
\usepackage{bbm,dsfont}
\usepackage{graphicx}
\usepackage{hyperref}
\usepackage{enumerate}
\usepackage{comment}
\usepackage{float}

\newtheorem{proposition}{Proposition}

\newtheorem{corollary}{Corollary}

\theoremstyle{definition}

\newtheorem{example}{Example}

\newcommand{\real}{\mathbb{R}} 
\newcommand{\nat}{\mathbb N} 
\newcommand{\half}{\tfrac{1}{2}} 

\newcommand{\hi}{\mathcal{H}} 
\newcommand{\lh}{\mathcal{L(H)}} 
\newcommand{\lsh}{\mathcal{L}_s(\mathcal{H})} 
\newcommand{\sh}{\mathcal{S(H)}} 
\newcommand{\eh}{\mathcal{E(H)}} 
\newcommand{\ip}[2]{\left\langle\,#1\,|\,#2\,\right\rangle} 
\newcommand{\no}[1]{\left\|#1\right\|} 
\newcommand{\tr}[1]{{\rm tr}\left[#1\right]} 
\newcommand{\id}{\mathbbm{1}} 

\newcommand{\Co}{\mathsf{C}}
\newcommand{\Ho}{\mathsf{H}}
\newcommand{\Mo}{\mathsf{M}}
\newcommand{\No}{\mathsf{N}}
\newcommand{\To}{\mathsf{T}}
\newcommand{\Jo}{\mathsf{J}}

\newcommand{\state}{\mathcal{S}} 
\newcommand{\effect}{\mathcal{E}} 
\newcommand{\obs}{\mathcal{M}} 
\newcommand{\vs}{\mathcal{V}} 

\newcommand{\conv}[1]{\mathrm{conv}\left( #1 \right)}
\newcommand{\aff}[1]{\mathrm{aff}\left( #1 \right)}

\newcommand{\lmax}{\lambda_{max}} 

\newcommand{\theory}{\mathcal{T}} 

\makeatletter
\renewcommand*\env@matrix[1][\arraystretch]{%
  \edef\arraystretch{#1}%
  \hskip -\arraycolsep
  \let\@ifnextchar\new@ifnextchar
  \array{*\c@MaxMatrixCols c}}
\makeatother

\makeatletter
\def\@settitle{\begin{center}%
  \baselineskip14\p@\relax
 
  \uppercasenonmath\@title
  \@title
  \ifx\@subtitle\@empty\else
     \\[1ex]\@subtitle
     
  \fi
  \end{center}%
}
\def\subtitle#1{\gdef\@subtitle{#1}}
\def\@subtitle{}
\makeatother

\begin{document}\setlength{\arraycolsep}{2pt}

\title[]{\textbf{Random access test as an identifier of nonclassicality}}

\author[Heinosaari]{Teiko Heinosaari}
\author[Lepp\"aj\"arvi]{Leevi Lepp\"aj\"arvi}

\email{Teiko Heinosaari: teiko.heinosaari@utu.fi}
\email{Leevi Lepp\"aj\"arvi: leevi.leppajarvi@gmail.com}

\address[Heinosaari]{(1) Department of Physics and Astronomy, University of Turku, Finland, (2)
Quantum algorithms and software, VTT Technical Research Centre of Finland Ltd}

\address[Lepp\"aj\"arvi]{RCQI, Institute of Physics, Slovak Academy of Sciences, Dúbravská cesta 9, 84511 Bratislava, Slovakia }

\begin{abstract}
Random access codes are an intriguing class of communication tasks that reveal an operational and quantitative difference between classical and quantum information processing. We formulate a natural generalization of random access codes and call them random access tests, defined for any finite collection of measurements in an arbitrary finite dimensional general probabilistic theory. These tests can be used to examine collective properties of collections of measurements. We show that the violation of a classical bound in a random access test is a signature of either measurement incompatibility or super information storability. The polygon theories are exhaustively analyzed and a critical difference between even and odd polygon theories is revealed.  
\end{abstract}

\maketitle

\begin{figure}[H]
\centering
\includegraphics[scale=0.09]{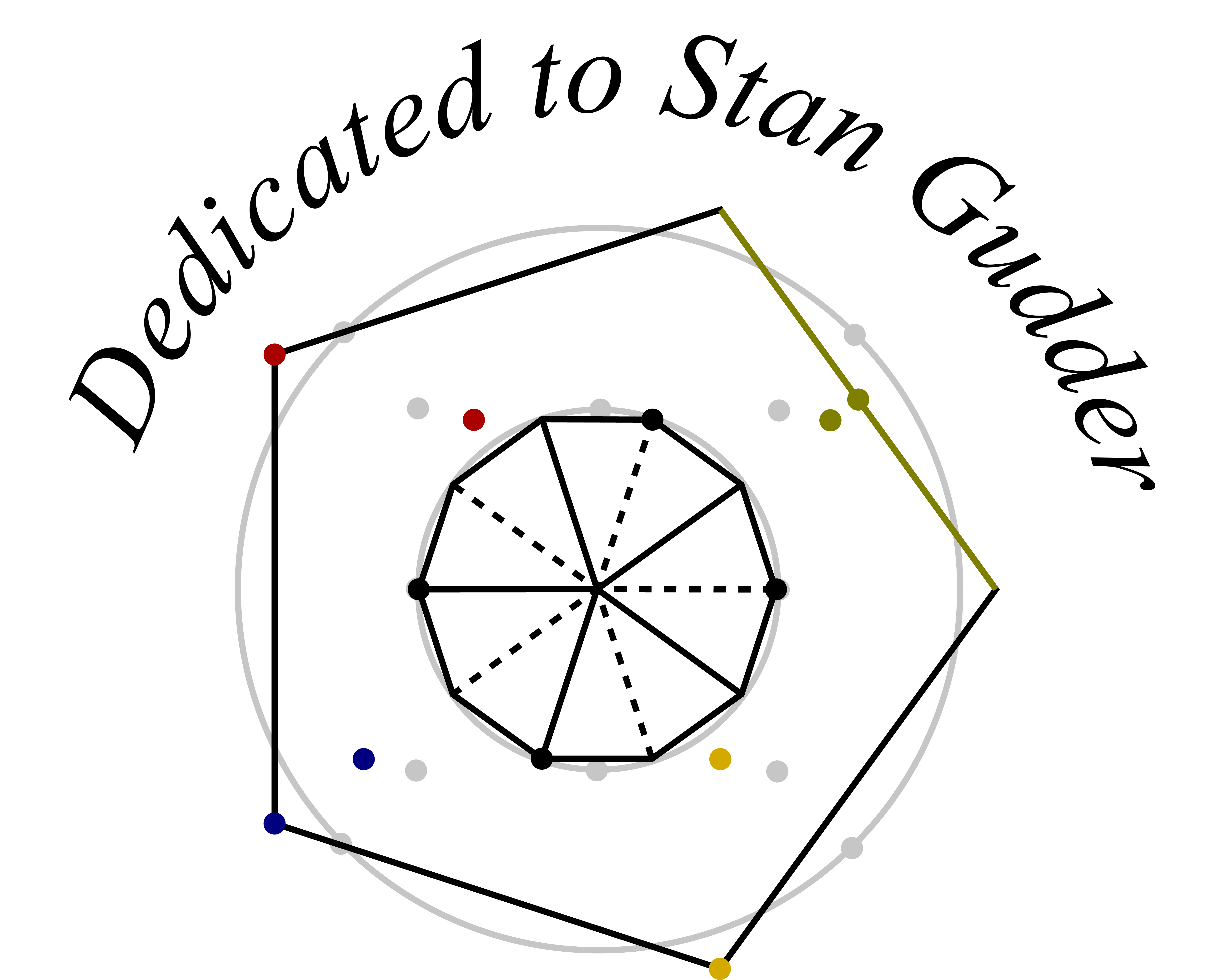} \ 
\end{figure}

\section{Introduction}

The central theoretical aim of quantum information processing is to understand how  quantum physics can be exploited in computing and communication. For example, in the celebrated superdense coding protocol quantum entanglement is used to communicate a certain number of bits of information by transmitting smaller number of qubits and sharing beforehand entangled qubits \cite{BeWi92}. In general, quantum protocols are superior to classical protocols in some information processing tasks, but not in all (such as in nonlocal computation \cite{LiPoShWi07} and some quantum guessing games \cite{AlBaBrAcGiPi10}). Every task showing a quantum advantage also reveals something about quantum theory itself, although the quantum resources behind the advantage may not always be so straightforward to identify.  

Random access codes (RACs) are simple communication protocols where a number of bits are encoded into a smaller number of bits and it is later randomly decided which bit should be decoded. This kind of tasks become interesting when compared to their quantum versions, known as quantum random access codes (QRACs), where a number of bits are encoded into a smaller number of qubits, or more generally, a quantum system with smaller dimension than the dimension of the classical encoding space \cite{Wiesner83, AmLeMaOz09}.  It is known that in many cases a quantum system is better than a classical system of the same operational dimension (i.e. with the same maximal number of perfectly distinguishable states), and quantum random access codes have been investigated from various different angles and generalized into different directions, see e.g. \cite{PaZu10, AgBoMiPa18,AmBaChKrRa19,DoMo21}. It has been e.g. shown that QRAC provides a robust self-testing of mutually unbiased bases, which are the unique quantum measurements giving the optimal performance in a particular QRAC \cite{FaKa19}. Further, it has been shown that to get any quantum advantage at all, one must use incompatible pair of measurements \cite{CaHeTo20}. In that way, QRAC can be used as a semi-device independent certification of quantum incompatibility. 

In the current investigation we adopt the approach started in \cite{CaHeTo20} and generalize random access codes to \emph{random access tests} (RATs) where the number of measurement outcomes and operational dimension of the communication medium are independent from each other (in RACs it is assumed that these are the same). 
We formulate random access tests in the framework of general probabilistic theories (GPTs), hence enabling us to compare the performances in different theories. 

We connect the performance of a collection of measurements in a random access test to the decoding power of their specific approximate joint measurement, which we call the \emph{harmonic approximate joint measurement}. This observation links the optimal performance of the RAT to the information storability of the full theory, a concept introduced in \cite{MaKi18}. Remarkably, it is found that the measurement incompatibility is not a necessary condition for a performance over the classical bound (as it is in quantum theory), but also a phenomenon called \emph{super information storability} can enable it. By super information storability we mean that the information storability is larger than the operational dimension of the theory; it's existence was recognized and studied in \cite{MaKi18}.

We make a detailed investigation of the performance of RATs in polygon theories. The optimal performances reveal a difference between even and odd polygon theories and, perhaps more surprisingly, also a finer division into different classes in both polygons. Our investigation indicates that the optimal performance in information processing tasks is a route to a deeper understanding of GPTs and their nonclassical features, thereby also to a better understanding of quantum theory.

The present investigation is organized as follows. In Sec. \ref{sec:GPT} we recall the needed definitions and machinery of general probabilistic theories. Sec. \ref{sec:IS} focuses on the concepts of decoding power of a measurement and information storability of a whole theory. In Sec. \ref{sec:APPROX} we define a specific kind of approximate joint measurement for any collection of measurements, which turns out to be a useful tool. In Sec. \ref{sec:RAT} we are finally ready to define random access tests and show how they link to the earlier concepts. Sec.  \ref{sec:MAX} connects the definite success of special random access tests to maximal incompatibility. In Sec. \ref{sec:POLY} we present a detailed study of random access tests in polygon state spaces and demonstrate how the maximal success probabilities separate different theories.

\section{General probabilistic theories}\label{sec:GPT}

General probabilistic theories constitutes a generalized framework for quantum and classical theories based on operational principles. In addition to quantum and classical theories, GPTs include countless toy theories where various operational features and tasks can be tested and considered in. This enables us to compare different theories to each other based on how these different features behave in different theories. In particular, by looking at the known nonclassical features of quantum theory (such as incompatibility \cite{BuHeScSt13,BaGaGhKa13}, steering \cite{StBu14,Banik15} and nonlocality \cite{PoRo94,JaGoBaBr11}) in this more general operational framework helps us understand what makes quantum theory special among all other possible theories. Furthermore, by looking at these features in the full scope of GPTs gives us insight on these features themselves which deepens our understanding about them and helps us make connections between different features.

GPTs are build around operational concepts such as \emph{preparations}, \emph{transformations} and \emph{measurements} which are used to describe a physical experiment. The preparation procedure involves preparing a (physical) system in a \emph{state} that contains the information about the system's properties. The set of possible states is described by a state space $\state$ which is taken to be a compact convex subset of a finite-dimensional vector space $\vs$. Whereas compactness and finite-dimensionality are technical assumptions which are often made to simplify the mathematical treatment of the theory, convexity follows from the possibility to have probabilistic mixtures of different preparation devices: if we prepare the system in a state $s_1\in \state$ with probability $p \in [0,1]$ or state $s_2 \in \state$ with probability $1-p \in [0,1]$ in different rounds of the experiment, then the prepared state is statistically described by the mixture $ps_1+(1-p)s_2$ and must thus be a valid state in $\state$. The extreme points of $\state$ are called pure and the set of pure states is described by $\state^{ext}$. If a state is not pure then it is called mixed.

In the popular ordered vector space formalism (see e.g. \cite{Lami17, Plavala21} for more details) the state space $\state$ is embedded as a compact convex base of a closed, generating proper cone $\vs_+$ in a finite-dimensional ordered vector space $\vs$. This means that $\vs_+$ is convex, it spans $\vs$, it satisfies $\vs_+ \cap - \vs_+ = \{0\}$ and that every element $x \in \vs_+ \setminus \{0\}$ has a unique base decomposition $x = \alpha s$, where $\alpha>0$ and $s \in \state$. In this case the state space can be expressed as
\begin{align*}
\state = \{ x \in \vs \, | \, x \geq 0, \ u(x)=1 \}, 
\end{align*}
where the the partial order $\geq$ in $\vs$ is the partial order induced by the proper cone $\vs_+$ defined in the usual way as $x \geq y$ if and only if $x-y \in \vs_+$, and $u$ is a order unit in the dual space $\vs^*$, or equivalently, a strictly positive functional on $\vs$. 

The measurement events are described by \emph{effects} which are taken to be affine functionals from $\state$ to the interval $[0,1]$; for an effect $e: \state \to [0,1]$ we interpret $e(s)$ as the probability that the measurement event corresponding to $e$ is detected when the system is in a state $s \in \state$. The affinity of the effects follows from the statistical correspondence between states and effects according to which the effects respect the convex structure of the states so that 
$$
e(ps_1+(1-p)s_2) = pe(s_1) + (1-p)e(s_2)
$$
for all $p\in [0,1]$ and $s_1,s_2 \in \state$. 

The set of effects on a state space $\state$ is called the effect space of $\state$ and it is denoted by $\effect(\state)$. Two distinguished effects are the zero effect $o$ and the unit effect $u$ for which $o(s)=0$ and $u(s)=1$ for all $s \in \state$. The effect space is clearly convex, and again the extreme elements of $\effect(\state)$ are called pure and others are mixed. The set of pure effects on a state space $\state$ is denoted by $\effect^{ext}(\state)$.

In the ordered vector space formalism we extent the effects to linear functionals and in this case we can depict the effect space as 
\begin{align*}
\effect(\state) = \vs^*_+ \cap (u-\vs^*_+) = \{ e \in \vs^* \, | \, o \leq e \leq u \},
\end{align*}
where the partial order on the dual space $\vs^*$ is induced by the dual cone $\vs^*_+ = \{f \in \vs^* \, | \, f(x) \geq 0 \ \forall x \in \vs_+ \}$.

An important class of effects is the set of indecomposable effects. Following \cite{KiNuIm10} we say that a nonzero effect $e\in \effect(\state)$ is indecomposable if the decomposition $e=e_1+e_2$ for some two nonzero effects $e_1, e_2 \in \effect(\state)$ implies that $e = c_1 e_1= c_2 e_2$ for some $c_1,c_2 > 0$. It was shown in \cite{KiNuIm10} that every effect can be expressed as a finite sum of indecomposable effect and that every indecomposable effect can be expressed as a positive multiple of a pure indecomposable effect. Geometrically indecomposable effects are exactly those that lie on the extreme rays of the dual cone $\vs^*_+$. We denote the set of indecomposable effects on $\state$ by $\effect_{ind}(\state)$ and the set of pure indecomposable effects by $\effect^{ext}_{ind}(\state)$.

A \emph{measurement} on a state space $\state$ with a finite number of outcomes is a mapping $\Mo: x \mapsto \Mo_x$ from a finite outcome set $\Omega$ to the set of effects $\effect(\state)$ such that $\sum_{x \in \Omega} \Mo_x(s) = 1$ for all $s \in \state$, or equivalently, $\sum_{x \in \Omega} \Mo_x = u$. We interpret $\Mo_x(s)$ as the probability that an outcome $x$ is obtained when the system in a state $s \in \state$ is measured with the measurement $\Mo$. We denote the set of all measurements on $\state$ by $\obs(\state)$.

The set of measurements with a fixed outcome set is convex, and we denote the set of all extreme measurements with any outcome set by $\obs^{ext}(\state)$. We say that an measurement is indecomposable if each of it's nonzero effect is indecomposable. We denote the set of indecomposable measurements on $\state$ by $\obs_{ind}(\state)$ and the set of indecomposable extreme measurements by $\obs^{ext}_{ind}(\state)$. A measurement $\To$ with an outcome set $\Omega$ is said to be trivial if it does not give any information about the measured state, i.e., it is of the form $\To_x = p_x u$ for all $x \in \Omega$ for some probability distribution $p:=(p_x)_{x \in \Omega}$ on $\Omega$.
 
There are two basic ways of forming new measurements: \emph{mixing} and \emph{post-processing}. First, as pointed out above, the set of measurements with a fixed set of outcomes is convex and thus we can make convex mixtures of measurements: if $\Mo^{(1)}, \ldots, \Mo^{(n)}$ are measurements with an outcome set $\Omega$ and $(p_i)_{i=1}^n$ is a probability distribution, then $\sum_{i=1}^n p_i \Mo^{(i)}$ is a measurement with effects $\sum_{i=1}^n p_i \Mo^{(i)}_x$ for all $x\in \Omega$. Clearly, those measurements that cannot be written as a nontrivial mixture are the extreme measurements. 

Second, we say that a measurement $\No$ with an outcome set $\Lambda$ is a post-processing of a measurement $\Mo$ with an outcome set $\Omega$ if there exists a stochastic matrix $\nu:= (\nu_{xy})_{x \in \Omega, y \in \Lambda}$, i.e., $\nu_{xy} \geq 0$ for all $x \in \Omega, y \in \Lambda$ and $\sum_{y \in \Lambda} \nu_{xy} =1$ for all $x \in \Omega$, such that 
\begin{align*}
\No_y = \sum_{x \in \Omega} \nu_{xy} \Mo_x
\end{align*}
for all $y \in \Lambda$. In this case we denote $\No = \nu \circ \Mo$. The post-processing relation defines a preorder on the set of measurements as follows: $\No \preceq \Mo$ if and only if $\No = \nu \circ \Mo$ for some stochastic matrix $\nu$. The set of maximal elements in $\obs(\state)$ with respect to the post-processing preorder is known to be exactly the set of  indecomposable measurements $\obs_{ind}(\state)$ \cite{MaMu90a, FiHeLe18}. 

\begin{example}[Quantum theory]
Let $\hi$ be a $d$-dimensional Hilbert space. We denote by $\lh$ the algebra of linear operators on $\hi$  and by $\lsh$ the real vector space of self-adjoint operators on $\hi$. The state space of a $d$-dimensional quantum theory is defined as
\begin{align*}
\sh = \{ \varrho \in \lsh \, | \, \varrho \geq O, \ \tr{\varrho}=1\},
\end{align*}
where $O$ is the zero-operator and the partial order is induced by the cone of positive semi-definite matrices according to which a self-adjoint matrix $A$ is positive semi-definite, $A \geq O$, if and only if $\ip{\varphi}{A \varphi} \geq 0$ for all $\varphi \in \hi$. The pure states are exactly the rank-1 projections on $\hi$.

The set of effects $\effect(\sh)$ can be shown (see e.g. \cite{MLQT12}) to be isomorphic to the set $\eh$ of self-adjoint operators bounded between $O$ and $\id$, where $\id$ is the identity operator on $\hi$, i.e.,
\begin{align*}
\effect(\sh) \simeq \eh := \{ E \in \lsh \, | \, O \leq E \leq \id \}.
\end{align*}
The pure effects then correspond to the projections on $\hi$ and the indecomposable effects to the rank-1 effect operators.

Measurements with finite number of outcomes on $\hi$ are described by positive operator-valued measures (POVMs), i.e., maps of the form $M: x \mapsto M(x)$ from a finite outcome set $\Omega$ to the set of effect operatos $\eh$ such that $\sum_{x \in \Omega} M(x) = \id$. The indecomposable POVMs are those whose all nonzero effects are rank-1 operators.
\end{example}

\section{Decoding power and information storability}\label{sec:IS}

\subsection{Base norms and order unit norms}

Let us start by introducing some more structure on GPTs that is needed in order to define decoding power and information storability (for more details on GPTs see e.g. \cite{Lami17, Plavala21}). Let $\state$ be a state space on an ordered vector space $\vs$. On the vector spaces $\vs$ and $\vs^*$ we can define two natural norms that are induced by the cones $\vs_+$ and $\vs_+^*$ respectively. We will introduce them next.

In the ordered vector space $\vs$ we have that $\state$ is a compact base of the closed, generating proper cone $\vs_+$ so that in particular every element $x \in \vs$ can be expressed as $x= \alpha y- \beta z$ for some $\alpha, \beta \geq 0$ and $y,z \in \state$. The \emph{base norm} $\no{\cdot}_\vs$ on $\vs$ is then defined as
\begin{align*}
\no{x}_\vs = \inf \{ \alpha + \beta \, | \, \ x=\alpha y- \beta z, \ \alpha, \beta \geq 0, \ y,z \in \state \}
\end{align*}
for all $x \in \vs$. It follows that if $x \in \vs_+$, then  $\no{x}_\vs = u(x)$. In particular we have that $\state = \{ x \in \vs_+ \, | \, \no{x}_\vs = 1\}$.

In the dual space $\vs^*$ we have the order unit $u \in \vs^*_+$, i.e., for every $f \in \vs^*$ there exists $\lambda > 0$ such that $f \leq \lambda u$, so that we can define the \emph{order unit norm} $\no{\cdot}_{\vs^*}$ on $\vs^*$ as
\begin{align*}
\no{f}_{\vs^*} = \inf \{ \lambda \geq 0 \, | \, - \lambda u \leq f \leq \lambda u\}
\end{align*}
for all $f \in \vs^*$. It follows that $\effect(\state) = \{ f \in \vs^*_+ \, | \, \no{f}_{\vs^* } \leq 1\}$. Furthermore, it can be shown that the base and the order unit norm are dual to each other, i.e.,
\begin{align*}
\no{x}_{\vs} = \sup_{ \no{f}_{\vs^*} \leq 1} |f(x)|, \qquad \no{f}_{\vs^*} = \sup_{ \no{x}_{\vs} \leq 1} |f(x)|.
\end{align*}
In particular, we can express the order unit norm on $\vs^*$ as the supremum norm over $\state$ so that
\begin{align} \label{eq:sup-norm}
\no{f}_{\vs^*} = \sup_{ s \in \state} |f(s)|
\end{align}
for all $f \in \vs^*$. Since $\state$ is compact, the supremum is always attained.

As we will mostly consider the properties of effects and measurements, we will be only using the norm $\no{ \cdot}_{\vs^*}$. For this reason, in order to simplify our notation, from now on we will write $\no{ \cdot}$ instead of $\no{ \cdot}_{\vs^*}$. We note that in particular for effects and other elements in $\vs^*_+$ the absolute values in Eq. \eqref{eq:sup-norm} can be removed, and thus we have that 
\begin{align*}
\no{e} = \max_{ s \in \state} e(s)
\end{align*}
for all $e \in \vs^*_+$.

\begin{example}[Quantum theory]
In the case of quantum theory $\sh$ we note that for an effect $e \in \effect(\sh)$ the norm $\no{e}$ corresponds to the operator norm $\no{ E }$ of the corresponding effect operator $E \in \eh$, which in the finite dimensional quantum theory equals with the maximal eigenvalue of $E$.
\end{example}

\subsection{Decoding power of a measurement}

For a measurement $\Mo \in \obs(\state)$ with an outcome set $\Omega$, we denote
\begin{align*}
\lmax(\Mo) := \sum_{x \in \Omega} \no{\Mo_x} \, .
\end{align*}
An operational interpretation of $\lmax(\Mo)$ is that two parties, a sender and a receiver, communicate by transferring physical systems where messages $1,\ldots,n$ are encoded in states $s_1,\ldots,s_n$. The receiver is bound to use $\Mo$ to decode the messages, but the sender can freely choose the states. For each outcome $x$ of $\Mo$, the sender hence chooses a state $s_x$ such that the correct inference is as likely as possible, i.e., $\Mo_x(s_x)$ is maximal. Assuming that the initial messages occur with uniform probability, the number $\lmax(\Mo) / n$ is the maximal probability for the receiver to infer the correct messages by using $\Mo$. Based on this operational motivation, we call $\lmax(\Mo)$ the \emph{decoding power} of $\Mo$.

In the following we show that $\lmax$ has certain monotonicity properties that makes it a reasonable quantification of the quality of measurements. In particular, we see that the decoding power function behaves in mixing and post-processing in the following way.

\begin{proposition}\label{prop:lmax-pp}
For any measurement $\Mo$ with an outcome set $\Omega$ and post-processing $\nu =(\nu_{xy})_{x \in \Omega, y \in  \Lambda}$, we have that
\begin{align*}
\lmax(\nu \circ \Mo) \leq \lmax(\Mo).
\end{align*}
\end{proposition}

\begin{proof}
With a direct calculation by using the triangular inequality and the absolute homogenity of $\no{\cdot }$ as well as the stochasticity of $\nu$ we see that
\begin{align*}
\lmax(\nu \circ \Mo) & = \sum_{y \in \Lambda} \no{(\nu \circ \Mo)_y}  
= \sum_{y \in \Lambda}  \no{ \sum_{x \in \Omega} \nu_{xy} \Mo_x } \\ 
&\leq \sum_{y \in \Lambda}   \sum_{x \in \Omega} \nu_{xy}  \no{ \Mo_x } =   \sum_{x \in \Omega} \left( \sum_{y \in \Lambda} \nu_{xy} \right)  \no{ \Mo_x } \\
& =  \sum_{x \in \Omega}  \no{ \Mo_x}
 = \lmax(\Mo) \, .
\end{align*}
\end{proof}

\begin{proposition}\label{prop:lmax-mix}
For any collection of measurements $\Mo^{(1)}, \ldots, \Mo^{(n)}$ with an outcome set $\Omega$ and probability distribution $p:=(p_i)_{i=1}^n$ we have that
\begin{align*}
\lmax\left(\sum_{i=1}^n p_i \Mo^{(i)}\right)  \leq \sum_{i=1}^n p_i \lmax(\Mo^{(i)}) \leq \max_{i\in \{1, \ldots, n\}} \lmax(\Mo^{(i)}).
\end{align*}
\end{proposition}

\begin{proof}
With a direct calculation by using the triangular inequality and the absolute homogenity of $\no{\cdot }$ as well as the normalization of $p$ we see that
\begin{align*}
\lmax\left(\sum_{i=1}^n p_i \Mo^{(i)}\right) & =  \sum_{x \in \Omega} \no{\sum_{i=1}^n p_i \Mo^{(i)}_x} \leq  \sum_{x \in \Omega}   \sum_{i=1}^n p_i \no{\Mo^{(i)}_x} \\
&=    \sum_{i=1}^n p_i \left( \sum_{x \in \Omega} \no{ \Mo^{(i)}_x} \right) = \sum_{i=1}^n p_i \lmax(\Mo^{(i)}) \\
& \leq \max_{i\in \{1, \ldots, n\}} \lmax(\Mo^{(i)}).
\end{align*}
\end{proof}

In the context of resource theories of measurements, the decoding power of a measurement is related to the robustness of measurements, defined as the minimal amount of noise needed to make a measurement trivial, studied in quantum theory e.g. in \cite{SkLi19, OsBi19} and in GPTs in \cite{Kuramochi20b}. 

\subsection{Information storability}

The previously defined decoding power is a quality of a single measurement. As we saw earlier, it is the maximal probability for a receiver to infer correct messages when the states used in encoding are optimized. We can also think of a scenario where the measurement is optimized over all possible measurements in the given theory. This quantity is known as the \emph{information storability} \cite{MaKi18} and for a theory with a state space $\state$ we denote it as
\begin{align}\label{eq:lmax-state}
\lmax(\obs(\state)) := \sup_{\Mo \in \obs(\state)} \lmax(\Mo) \, .
\end{align}
In the same way we can define the information storability for any subset  $\theory$ of measurements. This is particularly relevant when we study restrictions on measurements \cite{FiGuHeLe20}. For a subset $\theory \subseteq  \obs(\state)$ we denote the information storability of $\theory$ as 
\begin{align*}
\lmax(\theory) := \sup_{\Mo \in \theory} \lmax(\Mo) \, .
\end{align*}

In the cases where there is no risk of confusion we will use the simpler notation $\lmax(\state)$ for $\lmax(\obs(\state))$. 

\begin{example}[Quantum theory]\label{ex:quantum-is}
In a finite $d$-dimensional quantum theory for a POVM $M$ we have
\begin{align*}
\lmax(M) = \sum_x \no{M(x)} \leq \sum_x \tr{M(x)} = \tr{ \sum_x  M(x)} = \tr{\id} = d
\end{align*}
and we also see that this upper bound is reached whenever every operator $M(x)$ is rank-1.
\end{example}

Based on Example \ref{ex:quantum-is} one could presume that the information storability is always the same as the operational dimension of the theory, as is the case for quantum theory. The latter dimension is defined as the maximal number of perfectly distinguishable states. The operational dimension of a theory clearly is a lower bound for the information storability. However, the information storability can be larger than the operational dimension and we call this phenomenon \emph{super information storability}. This is the case e.g. in odd polygon theories (see Sec. \ref{sec:POLY}).

From Prop. \ref{prop:lmax-pp} and \ref{prop:lmax-mix} it follows that the supremum in Eq. \eqref{eq:lmax-state} is attained for extreme measurements that are maximal in the post-processing preorder. Thus, in a state space $\state$ we have that 
\begin{align}\label{eq:lmax-ind}
\lmax(\state) = \sup_{\Mo \in \obs(\state)} \lmax(\Mo) = \sup_{\Mo \in \obs^{ext}_{ind}(\state)} \lmax(\Mo).
\end{align}
We note that the set of extreme indecomposable measurements $\obs^{ext}_{ind}(\state)$ is exactly the set of extreme simulation irreducible measurements in \cite{FiHeLe18}. 

For state spaces with a particular structure we can show a simple way of calculating the information storability of the theory. In particular, we will use this result for polygon state spaces in Sec. \ref{sec:POLY}.

\begin{proposition}\label{prop:lambda-max-constant}
Suppose that there exists a state $s_0 \in \state$ such that $e(s_0) = f(s_0)=: \lambda_0$ for all $e,f \in \effect^{ext}_{ind}(\state)$. Then $\lmax(\Mo) = 1/ \lambda_0$ for all $\Mo \in \obs_{ind}(\state)$ and therefore $\lmax(\state)=1/ \lambda_0$.
\end{proposition}
\begin{proof}
Let $\Mo \in \obs_{ind}(\state)$ with an outcome set $\Omega$. Since each effect $\Mo_x$ is indecomposable, by \cite{KiNuIm10} we have that $\Mo_x = \alpha_x m_x$ for some $\alpha_x > 0$ and $m_x \in \effect^{ext}_{ind}(\state)$ for all $x \in \Omega$. Furthermore, by \cite{KiNuIm10}, for all $m_x$ there exists a pure state $s_x \in \state^{ext}$ such that $m_x(s_x)=1$ so that $\no{\Mo_x} = \alpha_x$ and $\lmax(\Mo) = \sum_{x \in \Omega} \alpha_x$.

From the normalization $\sum_{x \in \Omega} \Mo_x= u$ and the assumption that $e(s_0) = \lambda_0$ for all $e \in \effect^{ext}_{ind}(\state)$ it follows that 
\begin{align*}
\dfrac{1}{\lambda_0} = \dfrac{1}{\lambda_0} u(s_0) = \dfrac{1}{\lambda_0} \sum_{x \in \Omega} \alpha_x m_x(s_0) = \sum_{x \in \Omega} \alpha_x = \lmax(\Mo).
\end{align*}
The last claim follows from this and Eq. \eqref{eq:lmax-ind}.
\end{proof}

\section{Harmonic approximate joint measurements}\label{sec:APPROX}

A \emph{joint measurement} of measurements $\Mo^{(1)}, \ldots, \Mo^{(k)}$ with outcome sets $\Omega_1, \ldots, \Omega_k$ is a measurement $\Jo$ with the product outcome set $\Omega_1 \times \cdots \times \Omega_k$ and satisfying the marginal properties
\begin{align*} 
\sum_{x_2  \in \Omega_2} \cdots \sum_{x_k \in \Omega_k} \Jo_{x_1,\ldots, x_k} &= \Mo^{(1)}_{x_1} \quad  \forall x_1 \in \Omega_1, \\
\vdots \qquad \\
\sum_{x_1 \in \Omega_1} \cdots \sum_{x_{k-1} \in \Omega_{k-1}} \Jo_{x_1,\ldots, x_k} &= \Mo^{(k)}_{x_k}  \quad  \forall x_k \in \Omega_k\, .
\end{align*}
Measurements $\Mo^{(1)}, \ldots, \Mo^{(k)}$ are called \emph{compatible} if they have a joint measurement and otherwise they are \emph{incompatible}. 

In classical theories, i.e., in theories whose state spaces are simplices, all measurements are compatible whereas in every nonclassical GPTs there are pairs of measurements that are incompatible \cite{Plavala16,Kuramochi20}. However, even a set of incompatible measurements can be made incompatible if we allow for some amount of error. Most commonly this error is quantified by the amount of noise needed to be added to a set of incompatible set of measurements in order to make them compatible \cite{BuHeScSt13}. Formally, for each $(\lambda_1, \ldots, \lambda_k) \in [0,1]^k$ and every choice of trivial measurements $\To^{(1)}, \ldots, \To^{(k)}$ the measurements $\lambda_1 \Mo^{(1)} + (1-\lambda_1) \To^{(1)}, \ldots, \lambda_k \Mo^{(k)} + (1-\lambda_k) \To^{(k)}$ are considered to be \emph{noisy versions} of the measurements $\Mo^{(1)}, \ldots, \Mo^{(k)}$, where the amount of noise added to each observable $\Mo^{(i)}$ is characterized by the parameter $1-\lambda_i$. In this case we call a joint measurement $\tilde{\Jo}$ of any noisy versions of the measurements $\Mo^{(1)}, \ldots, \Mo^{(k)}$ an \emph{approximate joint measurement} of $\Mo^{(1)}, \ldots, \Mo^{(k)}$.

Following \cite{FiHeLe17} we can form a class of approximate joint measurements for $\Mo^{(1)}, \ldots, \Mo^{(k)}$ by fixing convex weights $\lambda_1,\ldots,\lambda_k$ (so that $\lambda_1+ \cdots+ \lambda_k=1$) and probability distributions $p^{(1)},\ldots,p^{(k)}$ on the sets $\Omega_2\times \cdots \times \Omega_k, \ldots, \Omega_1\times \cdots \times \Omega_{k-1}$ respectively and by setting
\begin{align}\label{eq:approx}
\tilde{\Jo}_{x_1,\ldots,x_k} = \lambda_1 p^{(1)}_{x_2,\ldots,x_k} \Mo^{(1)}_{x_1} + \cdots +\lambda_k p^{(k)}_{x_1,\ldots,x_{k-1}} \Mo^{(k)}_{x_k}
\end{align}
for all $(x_1, \ldots, x_k) \in \Omega_1 \times \cdots \times \Omega_k$.
The marginals of $\tilde{\Jo}$ are
\begin{align*} 
\sum_{x_2 \in \Omega_2} \cdots \sum_{x_k \in \Omega_k} \tilde{\Jo}_{x_1,\ldots, x_k} &= \lambda_1 \Mo^{(1)}_{x_1} + (1-\lambda_1) \To^{(1)}_{x_1}  \quad  \forall x_1 \in \Omega_1, \\
\vdots \qquad \\
\sum_{x_1 \in \Omega_1} \cdots \sum_{x_{k-1} \in \Omega_{k-1}} \tilde{\Jo}_{x_1,\ldots, x_k} &= \lambda_k \Mo^{(k)}_{x_k} + (1-\lambda_k) \To^{(k)}_{x_k} \quad \forall x_k \in \Omega_k ,
\end{align*}
where $\To^{(1)},\ldots,\To^{(k)}$ are trivial measurements.

In the previous construction we are free to choose the convex weights and the probability distributions. It turns out that a particular choice is useful for our following developments. We choose all probability distributions $p^{(1)},\ldots,p^{(k)}$ to be uniform distributions and the convex weights are chosen to be
\begin{align*}
\lambda_i = \frac{h(m_1,...,m_k)}{k\, m_i} \, ,
\end{align*}
where $h(m_1,\ldots,m_k)$ is the harmonic mean of the numbers $m_1,\ldots,m_k$ and $m_i$ is the number of outcomes of $\Mo^{(i)}$. We denote this specific measurement as $\Ho^{(1,\ldots,k)}$ and call it \emph{harmonic approximate joint measurement} of $\Mo^{(1)}, \ldots, \Mo^{(k)}$. Inserting the specific choices into Eq. \eqref{eq:approx} we observe that the harmonic approximate joint measurement can be written in the form
\begin{align}\label{eq:M-sum-gen}
\Ho^{(1,\ldots,k)}_{x_1,\ldots, x_k} = \frac{1}{\kappa(m_1,\ldots,m_k)}( \Mo^{(1)}_{x_1} + \cdots + \Mo^{(k)}_{x_k} ) \, ,
\end{align}
for all $(x_1, \ldots, x_k) \in \Omega_1 \times \cdots \times \Omega_k$ with
\begin{align*}
\kappa(m_1,\ldots,m_k) := \prod_i m_i \sum_i \frac{1}{m_i} \, .
\end{align*}

We remark that any approximate joint measurement of the form of Eq. \eqref{eq:approx} can be obtained from $\Mo^{(1)}, \ldots, \Mo^{(k)}$ by a suitable mixing and post-processing, hence can be simulated by using those measurements only \cite{FiHeLe18}. In this sense, they are all trivial approximate joint measurements and among these trivial approximate joint measurements our specific choice $\Ho^{(1,\ldots,k)}$ stands out by having a particularly symmetric form. Although trivial approximate joint measurements can be formed for any collection of measurements, we will show that the harmonic approximate joint measurement $\Ho^{(1,\ldots,k)}$ is related to the incompatibility of  $\Mo^{(1)}, \ldots, \Mo^{(k)}$ in an intriguing way. The link is explained in the following sections.

\section{Random access tests}\label{sec:RAT}

\subsection{Classical and quantum random access codes}

As a motivation for the later developments, we recall that in the $(n,d)$--\emph{random access code} (RAC), Alice is given $n$ input dits $\vec{x}=(x_1, \ldots, x_n) \in \{0, \ldots, d-1\}^n$ based on which she prepares one dit and sends it to Bob. Bob is then given a number $j \in \{1, \ldots, n\}$ and the task of Bob is to guess the corresponding dit $x_j$ of Alice. The temporal order is here important: Alice does not know $j$ and therefore cannot simply send $x_j$ to Bob. It is clear that, assuming all inputs are sampled uniformly, Bob will make errors. The performance depends on the strategy that Alice and Bob agree to follow. The choices of the inputs given to Alice and Bob are sampled with uniform probability and the average success probability quantifies the quality of their strategy. Generally, we denote by $\bar{P}^{(n,d)}_{c}$ the best average success probability that Alice and Bob can achieve with a coordinated strategy.  

For instance, Alice and Bob can agree that Alice always sends the value of the first dit. If Bob has to announce the value of the first dit, he makes no errors. This happens with the probability $1/n$. On the other hand, with the probaility $(n-1)/n$ Bob has to announce the value of some other dit in which case the information received from Alice does not help Bob and he has to make a random choice, thereby guessing the right value with the probability $1/d$. Therefore, with this strategy the average success probability is $(d+n-1)/(nd)$. It can be shown that in the case when $n=2$ this is also the optimal strategy so that $\bar{P}^{(2,d)}_c = 1/2(1+1/d)$ \cite{AmKrRa15,TaHaMaBo15}. 

In $(n,d)$--\emph{quantum random access code} (QRAC), Alice is given $n$ input dits $\vec{x}=(x_1, \ldots, x_n) \in \{0, \ldots, d-1\}^n$ based on which she prepares a $d$-dimensional quantum system into a state $\varrho_{\vec{x}} \in \mathcal{S}(\hi_d)$ and sends it to Bob. Bob is then given a number $j \in \{1, \ldots, n\}$ and the task of Bob is to guess the corresponding dit $x_j$ of Alice by performing a measurement on the state sent by Alice. If Bob performs a measurement described by a $d$-outcome POVM $M^{(j)}$, the average success probability of QRAC when the inputs are assumed to be uniformly distributed is given by
\begin{align*}
\dfrac{1}{n d^n} \sum_{\vec{x}} \tr{\varrho_{\vec{x}} (M^{(1)}(x_1) + \cdots + M^{(n)}(x_n) )}.
\end{align*}
The best average success probability $\bar{P}^{(n,d)}_{q}$ of a $(n,d)$-QRAC is then obtained by optimizing over the states and the measurements so that 
\begin{align*}
\bar{P}^{(n,d)}_{q} = \sup_{ M^{(1)}, \ldots, M^{(n)} \in \obs(\hi_d)} \dfrac{1}{n d^n} \sum_{\vec{x}} \no{M^{(1)}(x_1) + \cdots + M^{(n)}(x_n)},
\end{align*}
where the norm is the operator norm on $\mathcal{L}(\hi_d)$. In the case that $n=2$ it is known that $\bar{P}^{(2,d)}_q = 1/2(1+1/\sqrt{d}) > \bar{P}^{(2,d)}_c$ and that this bound can only be attained with rank-1 projective and mutually unbiased POVMs \cite{FaKa19}.

\subsection{Random access tests in GPTs} \label{subsec:gpt-rac}

In the following we formulate a class of tests that are similar to QRACs but are free from certain assumptions. Instead of restricting the number of measurement outcomes to match the (operational) dimension of the theory, we look at tests where the number of outcomes can be arbitrary. Furthermore, we look at these tests from the point of view of testing properties of collections of measurements in a given GPT. 

The tests are defined as follows: Let $\state$ be a state space and let $\Mo^{(1)},\ldots,\Mo^{(k)}$ be measurements on $\state$ with outcome sets $\Omega_1, \ldots, \Omega_k$ with $m_1,\ldots,m_k$ outcomes respectively. As in the QRAC setting,  Alice prepares a state $s_{\vec{x}} \in \state$, where $\vec{x}=(x_1,\ldots,x_k) \in \Omega_1 \times \cdots \times \Omega_k$. The task is the same as before: Bob is told an index $j\in \{1,\ldots,k\}$ and he should guess the label $x_j$ by performing a measurement $\Mo^{(j)}$ on the state $s_{\vec{x}}$. When the queries are uniformly distributed the average success probability of this \emph{random access test (RAT)} is then
\begin{align*}
\frac{1}{km_1\cdots m_k} \sum_{\vec{x}} \left(\Mo_{x_1}^{(1)} + \cdots + \Mo^{(k)}_{x_k} \right)(s_{\vec{x}})\, .
\end{align*}
When Bob is required to use the fixed measurements $\Mo^{(1)}, \ldots, \Mo^{(k)}$ and one optimizes over the states, the maximum average success probability becomes
\begin{align}\label{eq:P-def}
\bar{P}(\Mo^{(1)},\ldots,\Mo^{(k)}):=\frac{1}{km_1\cdots m_k} \sum_{\vec{x}} \no{ \Mo_{x_1}^{(1)} + \cdots + \Mo^{(k)}_{x_k}},
\end{align}
where the norm is now the order unit norm in the dual space $\vs^*$ of the ordered vector space $\vs$ in which the state space $\state$ is embedded. The number $\bar{P}(\Mo^{(1)},\ldots,\Mo^{(k)})$ is therefore the maximal success probability of the test for Alice and Bob if they are bound to use the measurements $\Mo^{(1)},\ldots,\Mo^{(k)}$ but free to choose the states.

A direct calculation reveals the following formula.

\begin{proposition}
\begin{align}\label{eq:connection}
 \bar{P}(\Mo^{(1)},\ldots,\Mo^{(k)})=\tfrac{1}{h(m_1,\ldots,m_k)} \, \lmax(\Ho^{(1,\ldots,k)})\, .
\end{align}
\end{proposition}

This observation shows that the two operationally accessible quantities $\bar{P}(\Mo^{(1)},\ldots,\Mo^{(k)})$ and $\lmax(\Ho^{(1,\ldots,k)})$ are connected in a simple way. However, let us emphasize that in real experiments one would find lower bounds for these quantities as they are defined via optimal states.  

From Eq. \eqref{eq:P-def} (or equivalently from Eq. \eqref{eq:connection}) one gets the following upper bound
\begin{align*}
\bar{P}(\Mo^{(1)}, \ldots, \Mo^{(k)}) \leq \frac{1}{k} \sum_{i=1}^k \frac{\lmax(\Mo^{(i)})}{m_i}\, .
\end{align*}
Thus, the maximum average success probability of a random access test is upper bounded by the weighted average of the decoding powers of the measurements used in the test. Clearly $\lmax(\Mo^{(i)}) \leq m_i$ for all $i \in \{1, \ldots, k\}$, where the equality holds only when $\Mo^{(i)}$ perfectly distinguishes $m_i$ (not necessarily different) states which is the only case when the bound is trivial. However, we note that in order to calculate this bound one has to be familiar with the description of the measurements $\Mo^{(1)},\ldots,\Mo^{(k)}$.

On the other hand, if we do not know the exact specification of the measurements but only know that $\Mo^{(1)},\ldots,\Mo^{(k)}$ belong to a set $\theory \subseteq \obs(\state)$ that is closed under post-processing and mixing, implying that $\Ho^{(1,\ldots,k)}$ belongs to $\theory$, then Eq. \eqref{eq:connection} gives an upper bound
\begin{align}\label{eq:universal-bound}
 \bar{P}(\Mo^{(1)},\ldots,\Mo^{(k)}) \leq \tfrac{1}{h(m_1,\ldots,m_k)} \lmax(\theory)\, .
\end{align}
 The minimal subset $\theory \subseteq \obs(\state)$ that includes the measurements $\Mo^{(1)}$, $\ldots$, $\Mo^{(k)}$ and that is closed with respect to mixing and post-processing is the simulation closure of the set $\{\Mo^{(i)}\}_{i=1}^k$ (see \cite{FiHeLe18}), i.e., a set formed of all possible mixtures and/or post-processings of the measurements $\Mo^{(1)}, \ldots, \Mo^{(k)}$. From Prop. \ref{prop:lmax-pp} and \ref{prop:lmax-mix} we see that in this case $\lmax(\theory) = \max_i \lmax(\Mo^{(i)})$.

We note that in the case when $\lmax(\theory)=m_1=\cdots= m_k=d$, where $d$ is the operational dimension of the theory (such as in the case of $(k,d)$--QRAC  where $\theory= \obs(\hi_d)$) the right hand side of Eq. \eqref{eq:universal-bound} is 1 and therefore it does not give a nontrivial bound. However, in different settings the bound can be nontrivial.

\begin{example}[$n$-tomic measurements]
Let $\theory$ consist of all measurements on $\state$ which can be simulated with measurements with $n$ or less outcomes, i.e., every measurement in $\theory$ is obtained as a mixture and/or post-processing of some set of $n$-outcome measurements. Measurements of this type are called $n$-tomic and have been considered e.g. in \cite{FiGuHeLe20, KlCa16, GuBaCuAc17, Huetal18}. It was shown in \cite{FiGuHeLe20} that $\lmax(\theory) \leq n$  so that Eq. \eqref{eq:universal-bound} gives $\bar{P}(\Mo^{(1)},\ldots,\Mo^{(k)}) \leq n/h(m_1,\ldots,m_k) $. For example, if $n=k=2$, then we obtain $\bar{P}(\Mo^{(1)},\Mo^{(2)}) \leq (m_1+m_2)/(m_1m_2) $ which is nontrivial for measurements $\Mo^{(1)},\Mo^{(2)} \in \theory$ with $m_1,m_2 >2$.
\end{example}

One can show that $\bar{P}(\Mo^{(1)}, \ldots, \Mo^{(k)})$ is a convex function of all of its arguments. To see this, let $\Mo^{(1)}, \ldots, \Mo^{(k)}$ be measurements and w.l.o.g. let $\Mo^{(1)} = \sum_{i=1}^l p_i \No^{(i)}$ for some measurements $\No^{(1)}, \ldots, \No^{(l)}$ and some probability distribution $(p_i)_{i=1}^l$. We see that
\begin{align*}
\bar{P}(\Mo^{(1)},\ldots, \Mo^{(k)})  &=  \frac{1}{k m_1 \cdots m_k} \sum_{\vec{x}} \no{ \Mo^{(1)}_{x_1} + \cdots + \Mo^{(k)}_{x_k}}\\
&=  \frac{1}{k m_1 \cdots m_k} \sum_{\vec{x}} \no{\sum_{i=1}^l p_i \No^{(i)}_{x_1} + \Mo^{(2)}_{x_2} + \cdots + \Mo^{(k)}_{x_k}}\\
&=  \frac{1}{k m_1 \cdots m_k} \sum_{\vec{x}} \no{\sum_{i=1}^l p_i \left( \No^{(i)}_{x_1} + \Mo^{(2)}_{x_2} + \cdots + \Mo^{(k)}_{x_k}\right)}\\
&\leq  \frac{1}{k m_1 \cdots m_k} \sum_{\vec{x}} \sum_{i=1}^l p_i \no{ \No^{(i)}_{x_1} + \Mo^{(2)}_{x_2} + \cdots + \Mo^{(k)}_{x_k} }\\
&= \sum_{i=1}^l p_i \bar{P}(\No^{(i)}, \Mo^{(2)}, \ldots, \Mo^{(k)})\, .
\end{align*}
Furthermore, one sees that
\begin{align*}
\bar{P}\left( \sum_i p_i \No^{(i)}, \Mo^{(2)}, \ldots, \Mo^{(k)} \right) &\leq \sum_{i=1}^l p_i \bar{P}(\No^{(i)}, \Mo^{(2)}, \ldots, \Mo^{(k)}) \\
& \leq \sum_{i=1}^l p_i \max_j \bar{P}(\No^{(j)}, \Mo^{(2)}, \ldots, \Mo^{(k)}) \\
&= \max_j \bar{P}(\No^{(j)}, \Mo^{(2)}, \ldots, \Mo^{(k)})\, .
\end{align*}
Thus, the success probability within the whole theory is maximized for extreme measurements.

\subsection{Upper bound for compatible pairs of measurements}

As was stated before, for the classical $(2,d)$-RAC it is known that the optimal success probability is $\bar{P}^{(2,d)}_c = \half \left(1+\frac{1}{d} \right)$. Generalizing the result from \cite{CaHeTo20}, let us consider random access tests with two measurements with $d$ outcomes on a state space $\state$, where the operational dimension of the theory is $d$. In such a GPT, there exists $d$ pure states $\{s_1, \ldots,s_d\} \in \state^{ext}$ and a $d$-outcome measurement $\Mo$ such that $\Mo_i(s_j)=\delta_{ij}$ for all $i,j \in \{1, \ldots, d\}$. We note that then $\no{\Mo_i}=1$ for all $i \in \{1, \ldots, d\}$ so that in particular we must have that $\no{\Mo_i + \Mo_i} =2$ and $\no{\Mo_i + \Mo_j} =1$ for all $i\neq j$ (since $\sum_i \Mo_i =u$). Hence, we see that
\begin{align*}
\bar{P}(\Mo,\Mo) &= \frac{1}{2d^2} \sum_{i,j=1}^d \no{\Mo_i + \Mo_j} = \frac{1}{d^2}  \sum_{i=1}^d \no{\Mo_i} +  \frac{1}{2d^2} \sum_{\substack{i,j=1 \\ i \neq j} }^d \no{\Mo_i + \Mo_j}  \\
&= \frac{1}{2d^2}(2 d + d(d-1))= \frac{1}{2} \left(1+\frac{1}{d} \right) \, , 
\end{align*}
so that the optimal success probability for the classical $(2,d)$--RAC can be always achieved in a theory with operational dimension $d$.

We see that the classical bound can be always achieved with the foolish strategy of choosing the same measurement for the random access test. For quantum theory it was shown in \cite{CaHeTo20} that this bound cannot be surpassed even if we allow for two different measurements that are compatible. However, in general we can actually show a bound for compatible measurements that in some cases differs from the classical one.

The following result is a generalization of Prop. 3 in \cite{CaHeTo20}. 

\begin{proposition}\label{prop:compatible}
Let $\Mo^{(1)}$ and $\Mo^{(2)}$ be two compatible measurements with $m_1$ and $m_2$ outcomes, respectively.
If $\Mo^{(1)}$ and $\Mo^{(2)}$ have a joint measurement belonging to $\theory \subseteq \obs(\state)$, then
\begin{align}\label{eq:compatible}
\bar{P}(\Mo^{(1)}, \Mo^{(2)}) \leq \frac{1}{2} \left( 1 + \frac{\lmax(\theory)}{m_1m_2} \right) \, .
\end{align}
If $\lmax(\theory)=m_1=m_2=d$, where $d$ is the operational dimension of the theory, then the right hand side is the classical bound for (2,d)--RAC.
\end{proposition}

\begin{proof}
Let  $\Jo \in \theory$ be a joint measurement of $\Mo^{(1)}$ and $\Mo^{(2)}$.
We have
\begin{align*}
&h(m_1,m_2) \cdot \bar{P}(\Mo^{(1)},\Mo^{(2)}) = \lmax(\Ho^{(1,2)}) = \sum_{x,y} \no{ \tfrac{1}{\kappa(m_1,m_2)}(\Mo^{(1)}_x+ \Mo^{(2)}_y) } \\
&= \frac{1}{\kappa(m_1,m_2)} \sum_{x,y} \no{ \sum_a  \Jo_{x,a} + \sum_b \Jo_{b,y} } \\
&= \frac{1}{\kappa(m_1,m_2)} \sum_{x,y} \no{\Jo_{x,y}+ \sum_{\substack{a\\ a \neq y}}  \Jo_{x,a} + \sum_{b} \Jo_{b,y} } \\
& \leq \frac{1}{\kappa(m_1,m_2)} \sum_{x,y} \no{\Jo_{x,y}}+ \frac{1}{\kappa(m_1,m_2)}\sum_{x,y}\no{ \sum_{\substack{a\\ a \neq y}}  \Jo_{x,a} +\sum_{b} \Jo_{b,y} } \\
& = \frac{1}{\kappa(m_1,m_2)} \lmax(\Jo)+ \frac{1}{\kappa(m_1,m_2)}\sum_{x,y}\no{ \sum_{\substack{a\\ a \neq y}}  \Jo_{x,a} +\sum_{b} \Jo_{b,y} }
\end{align*}
For the second summand, we observe that
\begin{align*}
0\leq \sum_{\substack{a\\ a \neq y}}  \Jo_{x,a} + \sum_{b} \Jo_{b,y} \leq \sum_{a,b} \Jo_{a,b} = u \, .
\end{align*}
We hence get
\begin{align*}
&h(m_1,m_2) \cdot \bar{P}(\Mo^{(1)},\Mo^{(2)}) \leq \frac{1}{\kappa(m_1,m_2)} \left( \lmax(\theory) + m_1m_2 \right) \, .
\end{align*}
Inserting the expressions of $h(m_1,m_2)$ and $\kappa(m_1,m_2)$ we arrive to Eq. \eqref{eq:compatible}.
\end{proof}

From the latter part of Prop. \ref{prop:compatible} it is clear that in the case of $(2,d)$--QRAC the bound given by Eq. \eqref{eq:compatible} for $\theory= \obs(\state)$ matches the optimal success probability of the classical $(2,d)$--RAC both for quantum and classical state spaces. In addition, in \cite{MaKi18} it was shown that $\lmax(\state)=d=2$ for all state spaces $\state$ that are point-symmetric (or centrally symmetric) so that also in this case the bound given by Eq. \eqref{eq:compatible} matches the optimal success probability of the classical $(2,2)$--RAC for two dichotomic measurements.

The bound given in Eq. \eqref{eq:compatible} can be used to detect incompatibility for pairs of measurements. To detect the incompatibility of any pair of measurements on a state space $\state$ we choose $\theory= \obs(\state)$ so that in terms of the harmonic approximate joint measurement we can express the previous result as follows:

\begin{corollary}
Let $\Mo^{(1)}$ and $\Mo^{(2)}$ be measurements on a state space $\state$ with $m_1$ and $m_2$ outcomes, respectively. If 
\begin{align}\label{eq:incompatible}
\lmax(\Ho^{(1,2)}) >  \frac{\lmax(\state)+m_1m_2}{m_1 +m_2},
\end{align}
then $\Mo^{(1)}$ and $\Mo^{(2)}$ are incompatible.
\end{corollary}

We remark that the previous incompatibility test becomes useless for large $m_1$ and $m_2$. Namely, $\lmax(\Ho^{(1,2)}) \leq  \lmax(\state)$ so that a necessary requirement that the inequality in Eq. \eqref{eq:incompatible} can hold for some $\Mo^{(1)}$ and $\Mo^{(2)}$ is that 
\begin{align}\label{eq:useful}
\lmax(\state) > \frac{m_1m_2}{m_1 + m_2 -1} \, .
\end{align}
In the usual QRAC scenario, i.e., when $m_1=m_2=d$, where $d$ is the operational dimension of the theory, from the fact that $\lmax(\state)\geq d$ it follows that Eq. \eqref{eq:useful} holds for all theories with $d>1$. Furthermore, if $m_1=m_2 = \lmax(\state)>1$, then Eq. \eqref{eq:useful} also holds. Some different types of examples as well as examples where the random access test does not detect incompatibility of some pairs of measurements in the quantum case have been presented in \cite{CaHeTo20}.

\section{Maximal incompatibility of dichotomic measurements} \label{sec:MAX}

There are various ways to quantify the level of incompatibility of two measurements \cite{DeFaKa19}. As introduced in \cite{BuHeScSt13}, the \emph{degree of incompatibility} $d(\Mo^{(1)},\Mo^{(2)})$ of two measurements $\Mo^{(1)}$ and $\Mo^{(2)}$ is defined as the maximal value $\lambda \in [0,1]$ such that the noisy versions $\lambda \Mo^{(1)} + (1-\lambda) \To^{(1)}$ and $\lambda \Mo^{(2)} + (1-\lambda) \To^{(2)}$ of $\Mo^{(1)}$ and $\Mo^{(2)}$ are compatible for some choices of trivial measurements $\To^{(1)}$ and $\To^{(2)}$. This quantity has a universal lower bound $d(\Mo^{(1)}, \Mo^{(2)}) \geq \half$, and two measurements $\Mo^{(1)}$ and $\Mo^{(2)}$ are called \emph{maximally incompatible} if $d(\Mo^{(1)}, \Mo^{(2)}) = \half$. For instance, the canonical position and momentum measurements in quantum theory are maximally incompatible \cite{HeScToZi14}. In the following we concentrate on maximal incompatibilty of two dichotomic measurements. In quantum theory it is known that no such pairs exist \cite{BuHeScSt13}.

We say that a dichotomic measurement $\Mo$ with outcomes $+$ and $-$ discriminates two sets of states $S_+, S_- \subset \state$ if $\Mo_+(s_+) = 1$ for all $s_+ \in S_+$ and $\Mo_+(s_-)=0$ for all $s_- \in S_-$. This concept allows now to formulate the following result.

\begin{proposition}\label{prop:2-2-maximal-prob}
Let $\Mo^{(1)}$ and $\Mo^{(2)}$ be two dichotomic measurements. Then $\bar{P}(\Mo^{(1)}, \Mo^{(2)}) =1$ if and only if there are four states $s_1,s_2,s_3,s_4 \in \state$ such that $\Mo^{(1)}$ discriminates $\{s_1,s_2\}$ and $\{s_3,s_4\}$ and $\Mo^{(2)}$ discriminates $\{s_1, s_4\}$ and $\{s_2,s_3\}$.
\end{proposition}
\begin{proof}
Let $\Mo^{(1)}$ and $\Mo^{(2)}$ be dichotomic measurements with effects $\Mo^{(1)}_+, \Mo^{(1)}_-$ and $\Mo^{(2)}_+, \Mo^{(2)}_-$ respectively. Since $\Mo^{(1)}_+ + \Mo^{(1)}_- = u$, $\Mo^{(2)}_+ + \Mo^{(2)}_- = u$ and $u(s) =1$ for all $s \in \state$, we can express the success probability $\bar{P}(\Mo^{(1)}, \Mo^{(2)})$ as follows:
\begin{align*}
\bar{P}(\Mo^{(1)}, \Mo^{(2)}) &= \frac{1}{8} \left[ \no{\Mo^{(1)}_+ + \Mo^{(2)}_+} + \no{\Mo^{(1)}_+ + \Mo^{(2)}_-} \right.   \\
& \quad \quad \quad \left. + \no{\Mo^{(1)}_- + \Mo^{(2)}_+}+ \no{\Mo^{(1)}_- + \Mo^{(2)}_-}\right]   \\
&= \frac{1}{8} \left[ \sup_{s_1 \in \state} (\Mo^{(1)}_+ + \Mo^{(2)}_+)(s_1) + \sup_{s_2 \in \state} (u+\Mo^{(1)}_+- \Mo^{(2)}_+)(s_2) \right.   \\
& \quad \quad \left. + \sup_{s_3 \in \state} (u-\Mo^{(1)}_+ + \Mo^{(2)}_+)(s_3)+ \sup_{s_4 \in \state} (2u-\Mo^{(1)}_+ - \Mo^{(2)}_+)(s_4)\right]   \\
&= \frac{1}{8} \left[4  + \sup_{s_1 \in \state}(\Mo^{(1)}_+ + \Mo^{(2)}_+)(s_1) + \sup_{s_2 \in \state}(\Mo^{(1)}_+ - \Mo^{(2)}_+)(s_2) \right.   \\
& \quad \quad \left. + \sup_{s_3 \in \state}(-\Mo^{(1)}_+ + \Mo^{(2)}_+)(s_3)+ \sup_{s_4 \in \state}(-\Mo^{(1)}_+ - \Mo^{(2)}_+)(s_4)\right]\, .
\end{align*}
Since $\Mo^{(1)}_+$ and $\Mo^{(2)}_+$ are effects, we clearly have that
\begin{align*}
& \sup_{s_1 \in \state}(\Mo^{(1)}_+ + \Mo^{(2)}_+)(s_1) \in [0,2], \quad \ \, \quad  \sup_{s_2 \in \state}(\Mo^{(1)}_+ - \Mo^{(2)}_+)(s_2) \in [-1,1], \\
 &\sup_{s_3 \in \state}(-\Mo^{(1)}_+ + \Mo^{(2)}_+)(s_3)\in [-1,1], \quad  \sup_{s_4 \in \state}(-\Mo^{(1)}_+ - \Mo^{(2)}_+)(s_4) \in [-2,0].
\end{align*}
Now it is clear that $\bar{P}(\Mo^{(1)}, \Mo^{(2)})=1$ if and only if there are four states $s_1,s_2,s_3,s_4$ such that 
\begin{align*}
(\Mo^{(1)}_+ + \Mo^{(2)}_+)(s_1) &=2, \quad (\Mo^{(1)}_+ - \Mo^{(2)}_+)(s_2) =1, \\
(-\Mo^{(1)}_+ + \Mo^{(2)}_+)(s_4) &=1, \quad  (-\Mo^{(1)}_+ - \Mo^{(2)}_+)(s_3) =0,
\end{align*}
which holds if and only if 
\begin{align} 
\Mo^{(1)}_+(s_1) = \Mo^{(1)}_+(s_2) = 1, \quad \Mo^{(1)}_+(s_3) = \Mo^{(1)}_+(s_4)=0, \label{eq:discrimination1} \\
\Mo^{(2)}_+(s_1) = \Mo^{(2)}_+(s_4) = 1, \quad \Mo^{(2)}_+(s_2) = \Mo^{(2)}_+(s_3)=0. \label{eq:discrimination2}
\end{align}
\end{proof}

In \cite{JePl17} it was shown that two dichotomic measurements $\Mo^{(1)}$ and $\Mo^{(2)}$ are maximally incompatible if and only if Eq. \eqref{eq:discrimination1} and \eqref{eq:discrimination2} hold for states $s_1,s_2,s_3,s_4 \in \state$ such that $\half (s_1 +s_3) = \half (s_2+s_4)$. Equivalently, \emph{two dichotomic measurements $\Mo^{(1)}$ and $\Mo^{(2)}$ are maximally incompatible if and only if there is an affine subspace $\mathcal{K} \subset \aff{\state}$ such that $\mathcal{F}=\mathcal{K} \cap \state$ is a parallelogram and $\Mo^{(1)}$ and $\Mo^{(2)}$ discriminate the opposite edges of $\mathcal{F}$}. By a parallelogram we mean a 2-dimensional convex body that is a convex hull of its four vertices $s_1, s_2, s_3, s_4$ such that $\half (s_1 +s_3) = \half (s_2+s_4)$. In this case measurements $\Mo^{(1)}$ and $\Mo^{(2)}$ discriminate the opposite edges of the parallelogram if $\Mo^{(1)}$ discriminates $\{s_1,s_2\}$ and $\{s_3,s_4\}$ and $\Mo^{(2)}$ discriminates $\{s_1, s_4\}$ and $\{s_2,s_3\}$. Thus, from Prop. \ref{prop:2-2-maximal-prob} and the aforementioned result from \cite{JePl17} we see that \emph{two maximally incompatible dichotomic measurements $\Mo^{(1)}$ and $\Mo^{(2)}$ satisfy $\bar{P}(\Mo^{(1)},\Mo^{(2)}) =1$}. In particular, we can conclude the following: 

\begin{corollary}
If $\bar{P}(\Mo^{(1)},\Mo^{(2)})<1$ for two dichotomic measurements $\Mo^{(1)}$ and $\Mo^{(2)}$, then they are not maximally incompatible.
\end{corollary}

On the other hand, in the case that the states in Prop. \ref{prop:2-2-maximal-prob} can be chosen to be affinely dependent, we can show also the inverse of the previous statement. To see this, let us consider the case of Prop. \ref{prop:2-2-maximal-prob} where the states $\{s_1,s_2,s_3,s_4\}$ are affinely dependent. First, let us note that from the discrimination of the sets $\{s_1,s_2\}$ and $\{s_3,s_4\}$ together with the discrimination of the sets $\{s_1,s_4\}$ and $\{s_2,s_3\}$ it follows that all the states are indeed different and must lie on the boundary of $\state$ so that in particular $\dim(\aff{\{s_1,s_2,s_3,s_4\}}) \geq 2$. Now, since $\{s_1,s_2,s_3,s_4\}$ are affinely dependent, we must have that $\dim(\aff{\{s_1,s_2,s_3,s_4\}}) =2$. In this case, in order for the discrimination of states to be possible, we must also have that 
\begin{align*}
\aff{\{s_1,s_2,s_3,s_4\}} & =\aff{\{s_1,s_2,s_3\}} = \aff{\{s_1,s_2,s_4\}}\\
& = \aff{\{s_1,s_3,s_4\}}= \aff{\{s_2,s_3,s_4\}},
\end{align*}
which in particular means that the states $s_2,s_3,s_4$ must form an affine basis of $\aff{\{s_1,s_2,s_3,s_4\}}$. Hence, we must have that $s_1 = \alpha_2 s_2 + \alpha_3 s_3 + \alpha_4 s_4$ for some $\alpha_2, \alpha_3, \alpha_4 \in \real$ such that $\alpha_2 + \alpha_3+ \alpha_4=1$. From Eq. \eqref{eq:discrimination1} and \eqref{eq:discrimination2} it follows that $\alpha_2=1$, $\alpha_3 =-1$ and $\alpha_4=1$ so that $\half(s_1+s_3)=\half(s_2+s_4)$. By the result of \cite{JePl17}, we then must have that $\Mo^{(1)}$ and $\Mo^{(2)}$ are maximally incompatible. We then arrive at the following result.

\begin{proposition}
Two dichotomic measurements $\Mo^{(1)}$ and $\Mo^{(2)}$ are maximally incompatible if and only if $\bar{P}(\Mo^{(1)}, \Mo^{(2)}) =1$, where the optimal states can be chosen to be affinely dependent.
\end{proposition}

In the first nonclassical polygon state space, the square state space $\state_4$, it is known that the measurements $\Mo^{(1)}$ and $\Mo^{(2)}$ defined by setting $\Mo^{(1)}_+ =e_1$ and $\Mo^{(2)}_+ =e_2$ (see Sec. \ref{sec:POLY} for more details) are maximally incompatible \cite{BuHeScSt13}. Indeed, one can check that then $\Mo^{(1)}$ discriminates $\{s_1,s_2\}$ and $\{s_3,s_4\}$ and $\Mo^{(2)}$ discriminates $\{s_1, s_4\}$ and $\{s_2,s_3\}$, where the states $s_1, \ldots, s_4$ are the affinely dependent pure states of the square state space. In similar fashion, one can construct maximally incompatible dichotomic measurements on theories whose state spaces are hypercubes of any dimension.

On the other hand, in the following example we demonstrate that the affine dependence of the optimal states is indeed necessary for measurements $\Mo^{(1)}$ and $\Mo^{(2)}$ with $\bar{P}(\Mo^{(1)}, \Mo^{(2)}) =1$ to be maximally incompatible.

\begin{example}[Tetrahedron]
Let $\state = \conv{\{s_1,s_2,s_3,s_4\}}$, where $s_1$, $s_2$, $s_2$, $s_3$ are the vertices of a regular tetrahedron. Since tetrahedron is a simplex, $\state$ is a classical state space with $d=4$ distinguishable pure states. Let $\Mo$ be the (unique) measurement that distinguishes the states $s_1,s_2,s_3,s_4$, i.e., $\Mo_i(s_j) = \delta_{ij}$ for all $i,j \in \{1,2,3,4\}$. We can define two dichotomic measurements $\Mo^{(1)}$ and $\Mo^{(2)}$ by setting $\Mo^{(1)}_+ = \Mo_1 + \Mo_2$ and $\Mo^{(2)}_+ = \Mo_1+\Mo_4$, and see that they satisfy Prop. \ref{prop:2-2-maximal-prob} only with the affinely independent pure states $s_1,s_2,s_3,s_4$. However, as $\state$ is a classical state space, all measurements on $\state$ are compatible so in particular $\Mo^{(1)}$ and $\Mo^{(2)}$ are not maximally incompatible.
\end{example}

\section{Random access tests in polygon state spaces}\label{sec:POLY}

\subsection{Polygon theories}

Folowing \cite{JaGoBaBr11} we define a regular $n$-sided polygon (or $n$-gon) state space $\state_n$ embedded in $\real^3$ as the convex hull of its $n$ extreme points
\begin{align*}
s_j=
\begin{pmatrix}
r_n \cos\left(\dfrac{2  j \pi}{n}\right) \\
r_n \sin\left(\dfrac{2 j \pi }{n}\right) \\
1
\end{pmatrix}, \quad j = 1,\ldots,n,
\end{align*}
where we have defined $r_n = \sqrt{\sec\left( \frac{\pi}{n}\right)}$. For $n=2$ the state space is a line segment which is the state space of the bit, i.e., the 2-dimensional classical system, and for $n=3$ the state space is a triangle which is the state space of the trit, i.e., the 3-dimensional classical system, while for all $n\geq 4$ we get nonclassical state spaces with the first one being the square state space.

Clearly, we now have the zero effect $o = (0, 0, 0)^T$ and the unit effect $u = (0, 0, 1)^T$. Depending on the parity of $n$, the effect space can have different structures (see \cite{FiHeLe18,HeLePl19} for details). Let us denote $s_0 = (0,0,1)^T$. For even $n$ the effect space $\effect(\state_n)$ has $n$ nontrivial extreme points,
\begin{align*}
e_k= \dfrac{1}{2}
\begin{pmatrix}
r_n \cos\left(\dfrac{(2k-1) \pi }{n}\right) \\
r_n \sin\left(\dfrac{(2k-1) \pi }{n}\right) \\
1
\end{pmatrix}, \quad k = 1,\ldots,n,
\end{align*}
so that $\effect(\state_n) = \conv{\{o,u,e_1, \ldots,e_n\}}$. 
All the nontrivial extreme effects lie on a single plane determined by those points $e$ such that $e(s_0)=1/2$.

In the case of odd $n$, the effect space has $2n$ nontrivial extreme effects,
\begin{align*}
g_k
= \dfrac{1}{1+r^2_n}
\begin{pmatrix}
r_n \cos\left(\dfrac{2k \pi }{n}\right) \\
r_n \sin\left(\dfrac{2k \pi }{n}\right) \\
1
\end{pmatrix}, \quad \quad f_k = u-g_k ,
\end{align*}
for $k = 1,\ldots,n$. Now $\effect(\state_n) = \conv{\{o,u, g_1, \ldots, g_n, f_1 ,\ldots, f_n\}}$ and the nontrivial effects are scattered on two different planes determined by all those points $g$ and $f$ such that $g(s_0) =  \frac{1}{1+r^2_n}$ and $f(s_0)  = \frac{r^2_n}{1+r^2_n}$. The first few polygons and their effect spaces are depicted in Fig. \ref{fig:polygons}.

\begin{figure}[t]
\centering
\includegraphics[width=\textwidth]{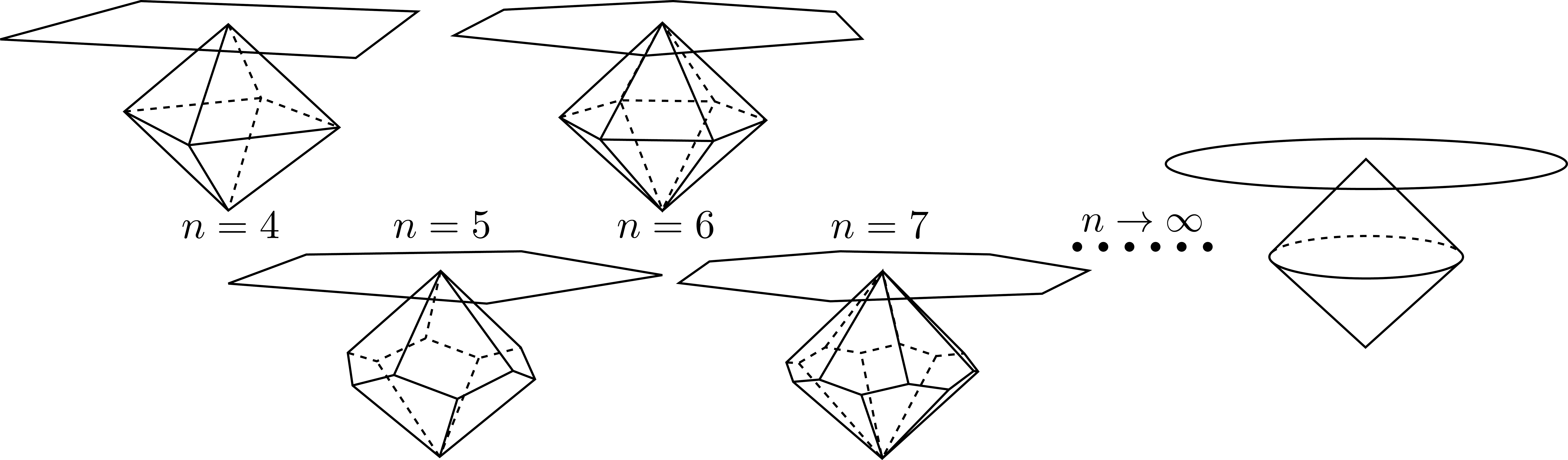}
\caption{\label{fig:polygons} Some of the first few even and odd polygon state spaces and their effect spaces. As the number of vertices increases, in both even and odd cases, the state and effect spaces start to resemble a disc and a circular bicone, respectively, which are the state and effect space of the real qubit system (which we discuss in detail later).  }
\end{figure}

Since $e \in \effect^{ext}_{ind}(\state_n)$ if and only if $e = e_k$ for some $k \in \{1, \ldots, n\}$ when $n$ is even, and $e = g_k$ for some $k \in \{1, \ldots, n\}$ when $n$ is odd, and because $e_k(s_0) = \half$ and $g_k(s_0) = 1/(1+r^2_n)$ for all $k \in \{1, \ldots,n\}$, we  can use Propositions \ref{prop:lmax-pp} and \ref{prop:lambda-max-constant} to make the following Corollary:

\begin{corollary}\label{cor:lmax-polygon}
$\lmax(\state_n) = 2$ when $n$ is even and $\lmax(\state_n)=1+\sec \left(\frac{\pi}{n}\right) >2$ when $n$ is odd. 
\end{corollary}

We note that for the nonclassical polygon state spaces $\state_n$ with $n\geq 4$, in both even and odd cases, the maximal number of distinguishable pure states is two, i.e., $d=2$ in both cases. On the other hand, clearly for the classical cases $n=2$ and $n=3$ we have $d=2$ and $d=3$, respectively. Thus, from the above Corollary we conclude that 
\begin{quote}
\emph{all nonclassical odd polygon state spaces have super information storability.} 
\end{quote}

Also, we note that the suitable classical reference for polygons is the bit ($n=d=2$) whereas the trit ($n=d=3$) does not share many of the features of the other polygons so that we often exclude the case $n=3$ when we consider the properties of the polygons as a whole. Therefore, in order to follow the (Q)RAC-like scenario, we will next focus on the simplest RATs on the nonclassical polygon state spaces $\state_n$ with $n\geq 4$, i.e., RATs with two measurements which both have $d=2$ outcomes. 

\subsection{Maximum success probability for compatible measurements}

From Cor. \ref{cor:lmax-polygon} we see that for even $n$, we have that $\lambda_{max}(\state_n)=2=d$ so that the bound given in Prop. \ref{prop:compatible} for two compatible measurements is exactly the classical bound for $(2,2)$--RAC. 
Thus, 
\begin{quote}
\emph{in the case of even polygon theories the classical bound can be achieved with compatible measurements but violation of the classical bound can only be achieved with incompatible measurements, just like in quantum theory}. 
\end{quote}

However, for odd $n$, we have that $\lambda_{max}(\state_n) = 1+r^2_n >2 =d$ and we can explicitly construct two compatible dichotomic measurements $\Mo^{(1)}$ and $\Mo^{(2)}$ such that $\bar{P}(\Mo^{(1)},\Mo^{(2)}) > \frac{3}{4}= \bar{P}^{(2,2)}_{c}$. Let us consider one of the indecomposable extreme measurements $\Co \in \obs^{ext}_{ind}(\state_n)$ with effects
\begin{align*}
\Co_1 = g_1 \, , \quad \Co_2= \half r^2_n \,  g_{\frac{n+1}{2}} \, , \quad \Co_3= \half r^2_n \, g_{\frac{n+3}{2}} \, .
\end{align*}
By denoting the two outcomes of $\Mo^{(1)}$ and $\Mo^{(2)}$ by $+$ and $-$, we take 
\begin{align*}
&\Mo^{(1)}_+=\Co_1 = g_1,  \quad  && \Mo^{(1)}_- = \Co_2+\Co_3 = \half r^2_n(g_{\frac{n+1}{2}}+g_{\frac{n+3}{2}}),  \\
&\Mo^{(2)}_+ = \Co_1 +\Co_2 = g_1 +\half r^2_n g_{\frac{n+1}{2}}, \quad  && \Mo^{(2)}_-= \Co_3 = \half r^2_n g_{\frac{n+3}{2}}. 
\end{align*}

\begin{figure}[t]
\centering
\includegraphics[width=0.7\textwidth]{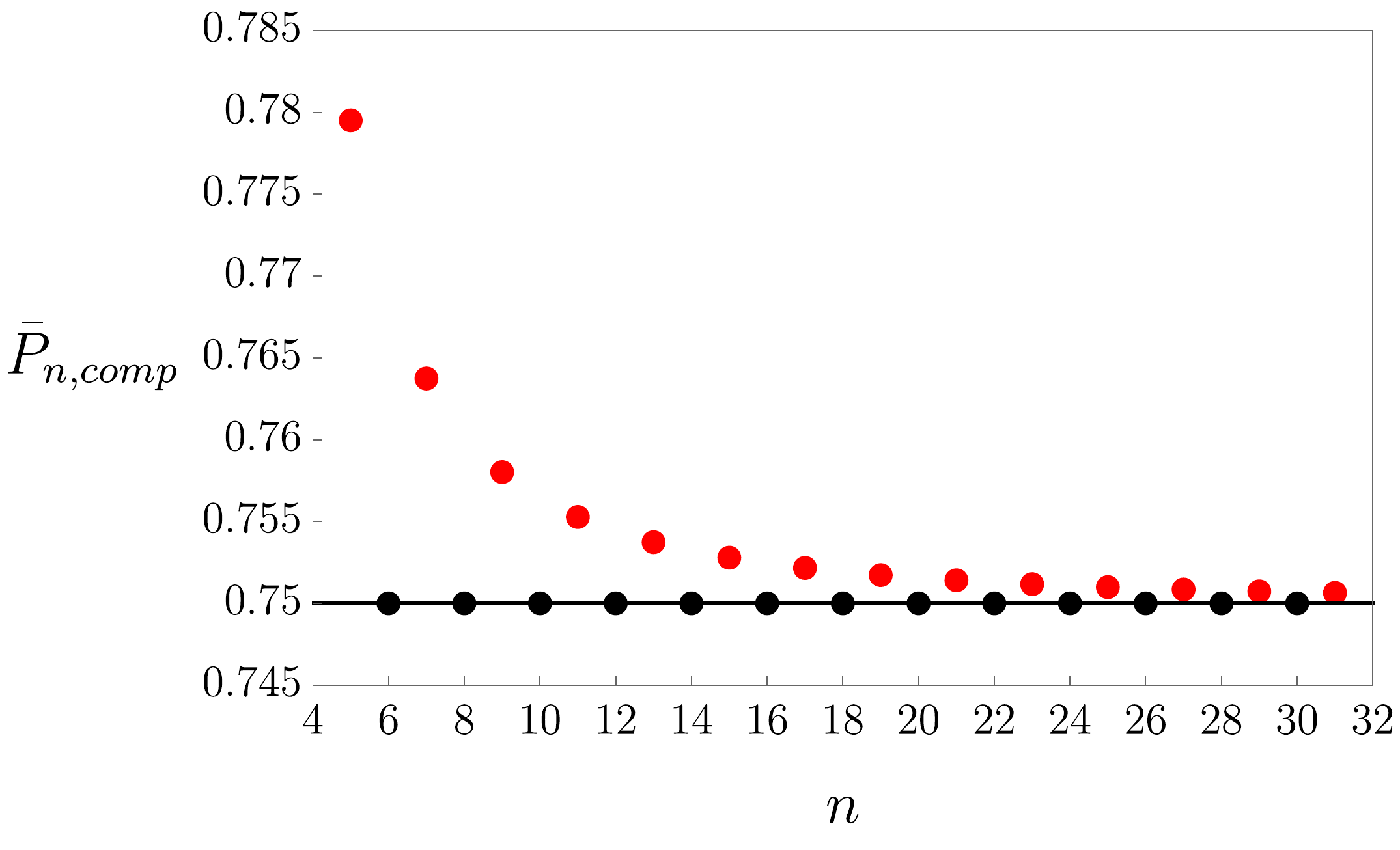}
\caption{\label{fig:polygon-comp-max} The maximum success probability $\bar{P}_{n, comp}$ of the random access test with two compatible dichotomic measurements on even (black dots) and odd (red dots) polygons as a function of the number of vertices $n$ of the polygon. The constant black line is the optimal success probability of the classical $(2,2)$--RAC.}
\end{figure}

Clearly $\Mo^{(1)}$ and $\Mo^{(2)}$ are compatible since they are both post-processings of $\Co$ and thus by Prop. \ref{prop:compatible} we have that
\begin{align*}
\bar{P}(\Mo^{(1)}, \Mo^{(2)}) \leq \frac{1}{2} \left( 1 + \frac{1+r^2_n}{4} \right).
\end{align*}
By using the states $s_1, s_{\frac{n+1}{2}}, s_{\frac{n+3}{2}}$ one can confirm that
\begin{align*}
\bar{P}(\Mo^{(1)}, \Mo^{(2)}) &= \frac{1}{8} \left( \no{\Mo^{(1)}_+ + \Mo^{(2)}_+} +\no{\Mo^{(1)}_+ + \Mo^{(2)}_-} \right. \\
& \quad \quad \quad \left. + \no{\Mo^{(1)}_- + \Mo^{(2)}_+}+ \no{\Mo^{(1)}_- + \Mo^{(2)}_-}\right) \\
&= \frac{1}{8} \left( \no{2 \Co_1 + \Co_2} + \no{\Co_1 + \Co_3} \right. \\
& \quad \quad \quad \left. + \no{\Co_1 + 2 \Co_2 + \Co_3} + \no{\Co_2 + 2 \Co_3} \right) \\
&\geq \frac{1}{8} \left[ (2 \Co_1 + \Co_2)(s_1) + (\Co_1 + \Co_3)(s_1) \right. \\
& \quad \quad  \left. + (\Co_1 + 2 \Co_2 + \Co_3)(s_{\frac{n+1}{2}}) + (\Co_2 + 2 \Co_3)(s_{\frac{n+3}{2}}) \right] \\
&= \frac{1}{2} \left( 1 + \frac{1+r^2_n}{4} \right).
\end{align*}
Therefore, by combining the above inequalities we see that the given states actually maximize the respective expressions and that 
\begin{align}\label{eq:odd-polygon-comp-max}
\bar{P}(\Mo^{(1)}, \Mo^{(2)})= \frac{1}{2} \left( 1 + \frac{1+r^2_n}{4} \right).
\end{align}
Since $\lambda_{max}(\state_n) =1+r^2_n > 2$, we conclude that 
\begin{align*}
\bar{P}(\Mo^{(1)},\Mo^{(2)}) > \frac{3}{4}= \bar{P}^{(2,2)}_{c} \, .
\end{align*}
Hence, 
\begin{quote}
\emph{in the case of odd polygon theories the classical bound can be surpassed with suitably chosen compatible measurements}. 
\end{quote}

We note that the right hand side of Eq. \eqref{eq:odd-polygon-comp-max} approaches the classical bound $\bar{P}^{(2,2)}_{c}=3/4$ as $n \to \infty$. The maximum success probability of a RAT for compatible pair of dichotomic measurements on polygon state spaces is illustrated in Fig. \ref{fig:polygon-comp-max}.

\subsection{Maximum success probability for incompatible measurements}

\subsubsection{Rebit state space} 
The state space of \emph{rebit}, or real qubit, is defined otherwise similarly to the qubit but the field of complex numbers is replaced with the field of real numbers. Thus, the `Bloch ball' of the qubit is replaced the `Bloch disc' so that rebit can be seen as a restriction of the qubit. Formally, the pure states and the nontrivial extreme effects of rebit are of the form $s_\theta = (\cos \theta, \sin \theta,1)^T$ and $e_\theta=\half (\cos \theta, \sin \theta,1)^T$ for any $\theta \in [0,2\pi)$ respectively. The zero and the unit effect are the same as in polygons, i.e., $o=(0,0,0)^T$ and $u=(0,0,1)^T$. In many ways polygons state spaces can be thought as discretized versions of the rebit state space. The following analysis of the RAT with two dichotomic measurements on the rebit will be useful in our later analysis of the polygon theories. 

We will explicitly show that the maximum success probabililty $\bar{P}$ of the RAT with two dichotomic measurements on the rebit is the same as in the corresponding $(2,2)$--QRAC in qubit, i.e., $1/2 \left( 1+ 1/\sqrt{2} \right)$. Hence, let us consider a RAT with two dichotomic measurements $\Mo^{(1)}$ and $\Mo^{(2)}$. We denote $e:= \Mo^{(1)}_+$ and $f := \Mo^{(2)}_+$ and the average success probability can then be written in the form
\begin{align}
\bar{P}(\Mo^{(1)}, \Mo^{(2)}) &= \frac{1}{8} \left[ \sup_{t_1 \in \state} (e+f)(t_1) + \sup_{t_2 \in \state}(u-e+f)(t_2) +\right. \nonumber \\
&\ \ \ \left. \sup_{t_3 \in \state}(u-e+u-f)(t_3) + \sup_{t_4 \in \state}(e+u-f)(t_4)\right]. \label{eq:P-2-2}
\end{align}
As was shown in Sec. \ref{subsec:gpt-rac}, the average success probability $\bar{P}$ is maximized for extreme measurements. On a state space $\state$ this means that Eq. \eqref{eq:P-2-2} is maximized for some effects $e,f \in \effect^{ext}(\state)$. Furthermore, the sums of effects in Eq. \eqref{eq:P-2-2} are maximized for pure states so that we can also choose the optimal values for $t_1, t_2, t_3, t_4$ to be pure if needed.

Due to the symmetry of the rebit system we can freely choose $\Mo^{(1)}_+=e=e_0=\half(1,0,1)^T$. Let us denote $f=e_\theta$  for some $\theta \in [0,2 \pi)$ and $t_i = s_{\varphi_i}$ for some $\varphi_i \in [0,2\pi)$ for all $i \in \{1,2,3,4\}$. The optimal success probability for a RAT with two dichotomic measurements on a rebit system then reads
\begin{align*}
\bar{P}&= \sup_{\theta \in [0,2\pi)} \frac{1}{8} \left[ \sup_{\varphi_1 \in [0,2\pi)} (e_0+e_\theta)(s_{\varphi_1}) + \sup_{\varphi_2 \in [0,2\pi)} (u-e_0+e_\theta)(s_{\varphi_2}) +\right. \\
&\ \ \ \left.   \sup_{\varphi_3 \in [0,2\pi)} (u-e_0+u-e_\theta)(s_{\varphi_3}) + \sup_{\varphi_4 \in [0,2\pi)} (e_0+u-e_\theta)(s_{\varphi_4})\right].
\end{align*}
By expanding the above expression and by using some trigonometric identities we can rewrite the above equation as
\begin{align*}
\bar{P}&= \sup_{\theta \in [0,2\pi)} \frac{1}{8} \left[ \sup_{\varphi_1 \in [0,2\pi)} \left(1+ \cos\left(\frac{\theta}{2} \right) \cos\left(\frac{\theta}{2}-\varphi_1 \right) \right)  \right. \\
& \quad \quad \quad \quad \quad \ \ + \sup_{\varphi_2 \in [0,2\pi)} \left(1- \sin\left(\frac{\theta}{2} \right) \sin\left(\frac{\theta}{2}-\varphi_2 \right) \right) \\
& \quad \quad \quad \quad \quad \ \ + \sup_{\varphi_3 \in [0,2\pi)} \left(1- \cos\left(\frac{\theta}{2} \right) \cos\left(\frac{\theta}{2}-\varphi_3 \right) \right) \\
 &  \left.  \quad \quad \quad \quad \quad \ \ + \sup_{\varphi_4 \in [0,2\pi)} \left(1+ \sin\left(\frac{\theta}{2} \right) \sin\left(\frac{\theta}{2}-\varphi_4 \right) \right) \right]. 
\end{align*}
We can get an upper bound for $\bar{P}$ by choosing the angles $\varphi_1, \varphi_2, \varphi_3, \varphi_4$ such that each of the inner supremums is bound above by either $1+|\cos(\theta/2)|$ or $1+|\sin(\theta/2)|$. After this the outer supremum can be calculated and we get the following upper bound:
\begin{align*}
\bar{P} &\leq \sup_{\theta \in [0,2\pi)} \frac{1}{8} \left[ 4 + 2 \left| \cos\left( \frac{\theta}{2} \right) \right| + 2 \left| \sin\left( \frac{\theta}{2} \right) \right|\right]  = \frac{1}{2} \left( 1+ \frac{1}{\sqrt{2}} \right)
\end{align*}

\begin{figure}[t]
\centering
\includegraphics[scale=0.25]{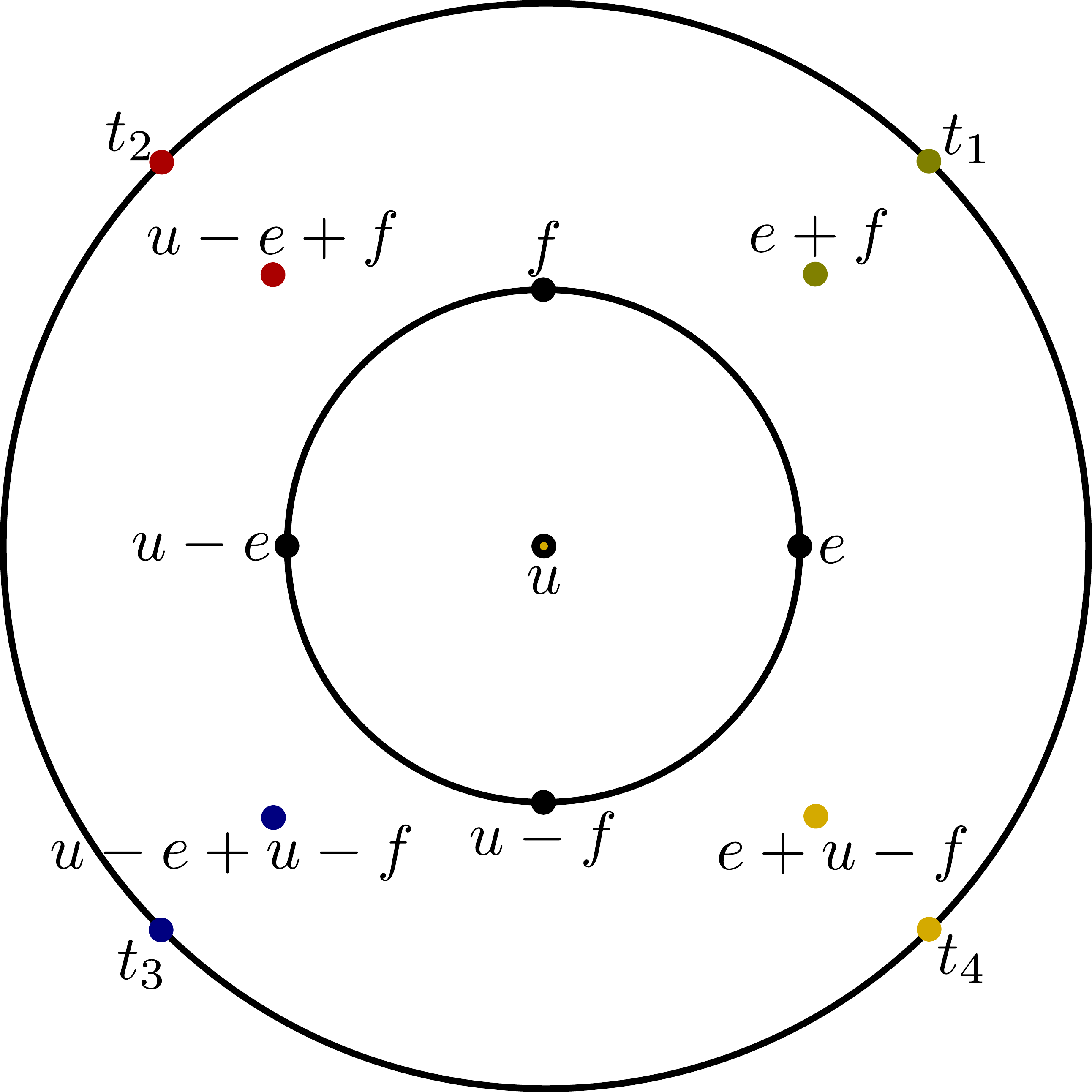}
\caption{\label{fig:rebit} The optimal effects $e= e_0$, $f= e_\theta$ and states $t_1 = s_{\varphi_1}$, $t_2 = s_{\varphi_2}$, $t_3 = s_{\varphi_3}$, $t_4 = s_{\varphi_4}$ with  $\theta=\pi/2$, $\varphi_1 = \pi/4$, $\varphi_2 = 3\pi/4$, $\varphi_3=5\pi/4$ and $\varphi_4= 7\pi/4$ for the maximal success probability (Eq. \eqref{eq:P-2-2}) for the random access test in a rebit system as viewed from the positive $z$-axis. }
\end{figure}

Furthermore, one can check that this bound is obtained with the following parameters: $\theta=\pi/2$, $\varphi_1 = \pi/4$, $\varphi_2 = 3\pi/4$, $\varphi_3=5\pi/4$ and $\varphi_4= 7\pi/4$. Thus, the optimal success probability of the random access test in rebit coincides with the corresponding optimal success probability in qubit.  The optimal effects and states (up to rotational symmetry) are depicted in Fig. \ref{fig:rebit}. We note that the optimal extreme effect $f$ aligns itself furthest away from both $e$ and $u-e$ along the semicircle between them (this is because we are optimizing both $e+f$ and $u-e+f$ at the same time) and the optimal states are uniquely determined by the sums $e+f$, $u-e+f$, $e+u-f$ and $u-e+u-f$ along the same directions in the $(x,y)$-projection as is seen in Fig. \ref{fig:rebit}.

\subsubsection{Maximum success probability for polygons} 
Based on the rebit system we can compare the behaviour of the maximal success probability of RATs on polygons. In many ways polygons can be thought as discretized versions of the rebit system and depending on the coarseness of the discretization the studied properties may look a bit different or similar to the properties of the rebit. This is also the case with the optimal success probability of the random access test. 

\begin{figure}[t!]
\centering
\includegraphics[scale=0.45]{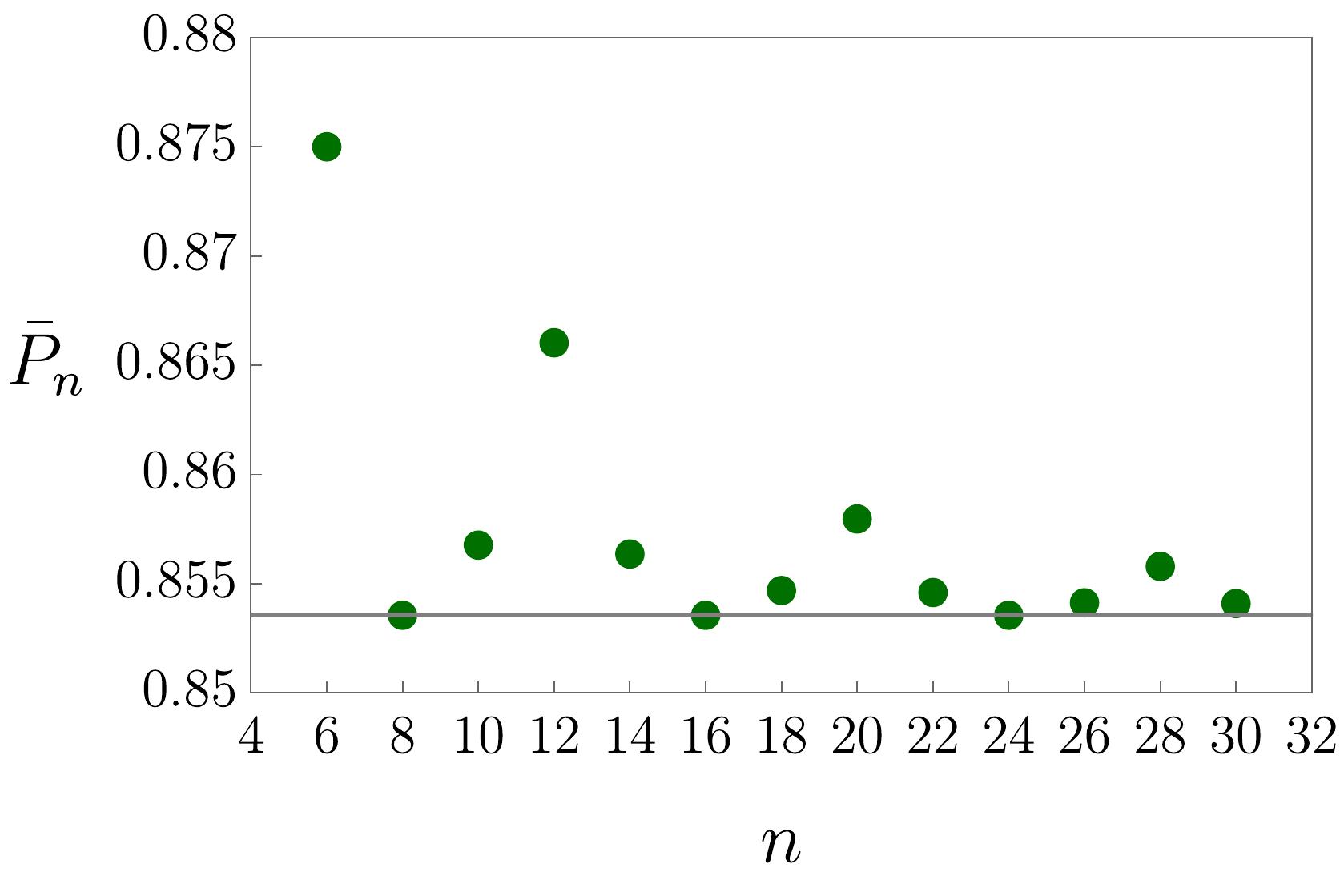}
\includegraphics[scale=0.45]{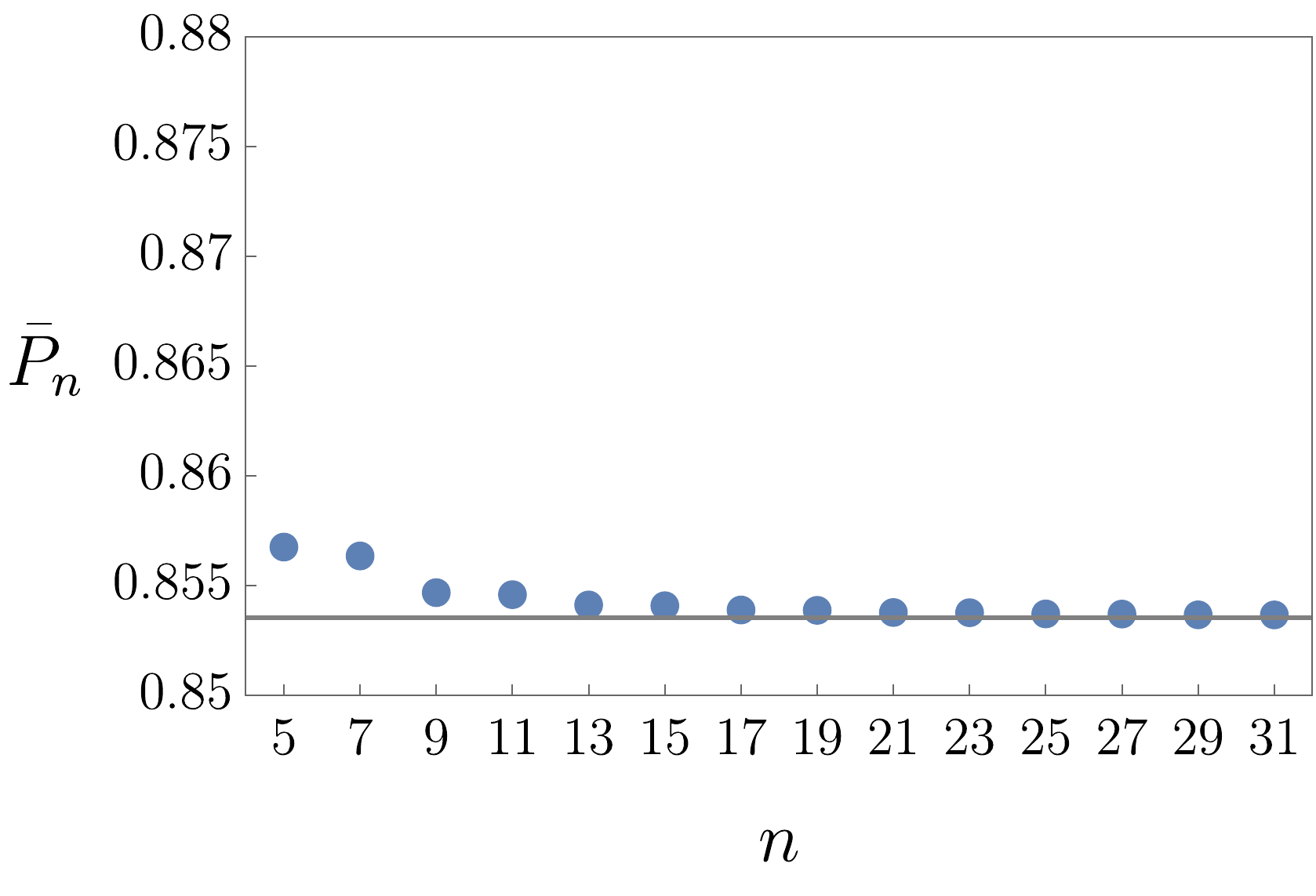}
\caption{\label{fig:polygon-max} The maximum success probability $\bar{P}_n$ of the random access test with two dichotomic measurements on even (above, green dots) and odd (below, blue dots) polygons as a function of the number of vertices $n$ of the polygon. The constant grey line is the maximum success probability of qubit (and rebit). From the even polygons we have excluded the case $n=4$ when $\bar{P}_n=1$.}
\end{figure}

As was established before, the maximum success probability for two dichotomic measurements on a regular polygon state space $\state_n$, denoted by $\bar{P}_n$, is maximized for some nontrivial extreme effects $e,f \in \effect^{ext}(\state_n)$ and pure states $t_1,t_2,t_3,t_4 \in \state^{ext}_n$. Because $\effect^{ext}(\state_n) = \{o,u,e_1, \ldots, e_n\}$ when $n$ is even and $\effect^{ext}(\state_n) = \{o,u,g_1, \ldots, g_n, f_1, \ldots, f_n\}$ when $n$ is odd, so that $\effect^{ext}(\state_n)$ (and $\state^{ext}_n$) is finite in both cases, it is easy to calculate the maximum value for $\bar{P}_n$ in the case of two dichotomic measurements for small values of $n$. The results are shown in Fig. \ref{fig:polygon-max} for polygons with up to $30$ vertices.

From Fig. \ref{fig:polygon-max} we observe at least three interesting points. First, the maximum success probability is at least as high as that of the qubit (and the rebit) in every polygon state space, and in many of them it is higher. Second, as one would expect, in both cases $\bar{P}_n$ seems to approach the maximum success probability of the qubit (and the rebit) as the number of vertices increase and the polygon starts to approximate the rebit system more closely geometrically. Third, the behaviour of $\bar{P}_n$ is drastically different for even and odd $n$: for even $n$ the maximum success probability $\bar{P}_n$ seems to oscillate with decreasing amplitude whereas for odd $n$ it seems that $\bar{P}_n$ approaches the qubit limit more monotonically. In the following we will analyse the results of Fig. \ref{fig:polygon-max} further by taking a closer look on the optimizing effects and states. We provide formulas for $\bar{P}_n$ in every $n$.

\subsubsection{Optimality results for the maximum success probability in the even polygons}
We will now look more closely on explaining Fig. \ref{fig:polygon-max}  for even polygons. We will prove that the maximum success probability $\bar{P}_n$ for two dichotomic measurements on even polygon state space $\state_n$ depends on $n$ as follows:
\begin{align}
&\bar{P}_n = \frac{1}{2} \left( 1+ \frac{\sec\left( \frac{\pi}{n} \right)}{\sqrt{2}} \right) \quad  \textrm{$n=4m$ for odd $m \in \nat$ }\label{eq:n=4m-odd} \\
& \bar{P}_n = \frac{1}{2} \left( 1+ \frac{1}{\sqrt{2}} \right) \quad  \textrm{$n=4m$ for even $m \in \nat$}\label{eq:n=4m-even} \\
& \bar{P}_n = \frac{1}{4} \left[ 2 +  r_n^2 \cos\left(\frac{m\pi}{n} \right) +    \sin\left(\frac{m\pi}{n} \right) \right] \quad \textrm{$n=4m+2$ for odd $m \in \nat$}\label{eq:n=4m+2-odd} \\
& \bar{P}_n = \frac{1}{4} \left[ 2 +   \cos\left(\frac{m\pi}{n} \right) +   r_n^2 \sin\left(\frac{m\pi}{n} \right) \right] \quad \textrm{$n=4m+2$ for even $m \in \nat$}\label{eq:n=4m+2-even}
\end{align}
(In the last two expressions $r_n = \sqrt{\sec\left( \frac{\pi}{n}\right)}$ as before.) Thus, we will see that in the cases when $n=4m+2$, where $m$ is odd or even, or $n=4m$, where $m$ is even, the maximum success probability is strictly larger than that of the qubit and the rebit, whereas for $n=4m$, where $m$ is even, it is exactly the same as in qubit. In all cases the limit $n \to \infty$ matches the qubit value.

We start by deriving an upper for the maximum success probability $\bar{P}_n$ similarly to how we did in the case of rebit. Let us consider the RAT for two dichotomic measurements as stated in Eq. \eqref{eq:P-2-2}. The polygons are symmetrical in the sense that we can fix the first extreme effect $e$ to be any of the nontrivial extreme effects and for even polygons we choose $e=e_1$. Furthermore, we have that $f = e_k$ for some $k \in \{1, \ldots, n\}$ and $t_i = s_{j_i}$ for some $j_i \in \{1, \ldots, n\}$ for all $i \in \{1,2,3,4\}$. With this notation we can write $\bar{P}_n$ for even polygons as
\begin{align}\label{eq:Pn-poly}
\begin{split}
\bar{P}_n&= \sup_{k \in \{1, \ldots, n\}} \frac{1}{8} \left[ \sup_{j_1 \in \{1, \ldots, n\}} (e_1+e_k)(s_{j_1}) + \sup_{j_2 \in \{1, \ldots, n\}} (u-e_1+e_k)(s_{j_2}) \right. \\
&\ \ \ \left.   + \sup_{j_3 \in \{1, \ldots, n\}} (u-e_1+u-e_k)(s_{j_3}) + \sup_{j_4 \in \{1, \ldots, n\}} (e_1+u-e_k)(s_{j_4})\right].
\end{split}
\end{align}
First we note that we can restrict $k \in \{1, \ldots, n/2\}$ since otherwise we can just take $f=u-e_k = e_{k+n/2}$ instead of $f=e_k$. By expanding the previous expression and by using some trigonometric identities we can rewrite the previous equation as
\begin{align*}
\bar{P}_n&= \sup_{k \in \{1, \ldots, n/2\}} \frac{1}{8} \left[\sup_{j_1 \in \{1, \ldots, n\}} \left(1+ r_n^2 \cos\left(\frac{(k-1)\pi}{n} \right) \cos\left(\frac{(k-2j_1)\pi}{n} \right) \right)  \right. \\
& \quad \quad \quad \quad \quad \ \ + \sup_{j_2 \in \{1, \ldots, n\}} \left(1- r_n^2 \sin\left(\frac{(k-1)\pi}{n} \right) \sin\left(\frac{(k-2j_2)\pi}{n} \right) \right) \\
& \quad \quad \quad \quad \quad \ \ + \sup_{j_3 \in \{1, \ldots, n\}} \left(1- r_n^2 \cos\left(\frac{(k-1)\pi}{n} \right) \cos\left(\frac{(k-2j_3)\pi}{n} \right) \right) \\
 &  \left.  \quad \quad \quad \quad \quad \ \ + \sup_{j_4 \in \{1, \ldots, n\}} \left(1+ r_n^2 \sin\left(\frac{(k-1)\pi}{n} \right) \sin\left(\frac{(k-2j_4)\pi}{n} \right) \right) \right]. 
\end{align*}

Analogously to the rebit, we can upper bound the inner supremums by the terms $1+r_n^2 |\cos((k-1)\pi/n)|$ and $1+r_n^2 |\sin((k-1)\pi/n)|$ from which we can omit the absolute values since $\cos((k-1)\pi/n) \geq 0$ and $\sin((k-1)\pi/n) \geq 0$ for all $k \in \{1, \ldots, n/2\}$. For the remaining (outer) supremum we can use the upper bound $\cos((k-1)\pi/n)+\sin((k-1)\pi/n) \leq \sqrt{2}$ so that in the end we get the following upper bound for $\bar{P}_n$ for all even polygons:
\begin{align}
\bar{P}_n &\leq \sup_{k \in \{1, \ldots,n/2\}} \frac{1}{8} \left[ 4 + 2 r_n^2 \cos\left(\frac{(k-1)\pi}{n} \right) + 2 r_n^2 \sin\left(\frac{(k-1)\pi}{n} \right) \right]  \nonumber \\
&\leq \frac{1}{2} \left( 1+ \frac{r_n^2}{\sqrt{2}} \right) = \frac{1}{2} \left( 1+ \frac{\sec\left(\frac{\pi}{n}\right)}{\sqrt{2}} \right). \label{eq:P_n-even-bound}
\end{align}

One can confirm that the above bound is attained when we take $k=1+n/4$, $j_1=k/2$, $j_2=k/2+n/4$, $j_3=k/2+n/2$ and $j_4=k/2+3n/4$. However, since $k,j_1,j_2,j_3,j_4$ must be integers, this upper bound can be attained only in the case when $n=4m$ for some $m\in \nat$ (so that $k$ is an integer) and $m$ is odd (so that $k=1+m$ is even and $j_1,j_2,j_3,j_4$ are integers). Thus, in this case we have that $f=e_{m+1}$, $t_1 = s_{\frac{m+1}{2}}$, $t_2 = s_{\frac{3m+1}{2}}$, $t_3 = s_{\frac{5m+1}{2}}$ and $t_4 = s_{\frac{7m+1}{2}}$.

From the previous result it becomes immediate that we must consider different cases also within the even polygons. Let us next consider the case when $n=4m$ but $m$ is even. As we saw above, we cannot saturate the previous bound for $\bar{P}_n$ in this case because in this case the optimizing values for $j_i$'s are not integers. However, the expressions for the inner supremums which we are upper bounding, namely $\cos((k-2j_1)\pi/n)$, $-\sin((k-2j_2)\pi/n)$, $\sin((k-2j_3)\pi/n)$ and $-\cos((k-2j_4)\pi/n)$, are discrete and of simple form so that we know that even if we cannot attain the optimal value $1$ for these expressions with the optimal parameters, the actual supremums are attained with parameters close to the the ones presented above. 

Let us consider separately two different cases when $n=4m$ and $m$ is even. First, let us consider the case that $k$ is even so that the optimal values for $j_i$'s are integers and the first inequality in Eq. \eqref{eq:P_n-even-bound} is saturated, i.e., the inner supremums have the same optimal values as above. Again, in this case they are either of the form $1+r_n^2 \cos((k-1)\pi/n)$ or $1+r_n^2 \sin((k-1)\pi/n)$ and they are attained with $j_1=k/2$, $j_2=k/2+m$, $j_3=k/2+2m$ and $j_4=k/2+3m$. However, the second inequality is saturated, i.e., the outer supremum attains the previous optimal value, only when $k=m+1$ which would make $k$ odd since $m$ is even. Therefore we cannot attain the previous bound in this case. Instead, for the outer supremum we must maximize $\cos((k-1)\pi/n)+\sin((k-1)\pi/n)$ for all even values of $k \in \{1, \ldots, n/2\}$. From the form of this expression we see that the supremum must be attained with the closest even integer value to $m+1$, i.e., either with $k=m$ or $k=m+2$. We can verify that then for both of these values of $k$ we have that $\cos((k-1)\pi/n)+\sin((k-1)\pi/n) = \sqrt{2}/r_n^2$ and hence $\bar{P}_n(e,f) = 1/2(1+1/\sqrt{2})$ with the optimizing states $t_1 = s_{\frac{m}{2}}$, $t_2 = s_{\frac{3m}{2}}$, $t_3 = s_{\frac{5m}{2}}$ and $t_4 = s_{\frac{7m}{2}}$ for the optimizing effect $f= e_m$, and $t_1 = s_{\frac{m}{2}+1}$, $t_2 = s_{\frac{3m}{2}+1}$, $t_3 = s_{\frac{5m}{2}+1}$ and $t_4 = s_{\frac{7m}{2}+1}$ for the optimizing effect $f= e_{m+2}$.

Second, let us assume that $k$ is odd so that the first inequality in Eq. \eqref{eq:P_n-even-bound} is not saturated, i.e., that the previously presented optimal values for $j_i$'s are not integers. In this case the actual optimizing values for the inner supremums must then be the closest integers to the values considered before. Thus, the inner supremums must be attained with the following values: $j_{\pm 1} = (k\pm 1)/2$, $j_{\pm 2} = (k\pm 1)/2+m$, $j_{\pm 3} = (k\pm 1)/2+2m$ and $j_{\pm 4} = (k\pm 1)/2+3m$. It is straightforward to verify that in both cases, i.e., when one uses either $j_{+i}$'s or $j_{-i}$'s, the inner supremums take the values $1+\cos((k-1)\pi/n)$ or $1+\sin((k-1)\pi/n)$ depending on which expression we are evaluating. Now the outer supermum can attain the same optimal value as before with the parameter value $k=m+1$ (which is odd because $m$ is even). Finally we obtain that also in this case $\bar{P}_n(e,f) = 1/2(1+1/\sqrt{2})$ with the effect $f=e_{m+1}$ and optimizing states $t_1 = s_{\frac{m}{2}}$ or $t_1 = s_{\frac{m}{2}+1}$, $t_2 = s_{\frac{3m}{2}}$ or $t_2 = s_{\frac{3m}{2}+1}$, $t_3 = s_{\frac{5m}{2}}$ or $t_3 = s_{\frac{5m}{2}+1}$ and $t_4 = s_{\frac{7m}{2}}$ or $t_4 = s_{\frac{7m}{2}+1}$. 

To conclude, we have shown that the maximum success probability for an even polygon with $n=4m$ for some $m \in \nat$ are given by Eqs. \eqref{eq:n=4m-odd} and \eqref{eq:n=4m-even}. All the optimal effects and states are explicitly expressed in Table \ref{table1}.

\begin{table}
\centering
\begin{tabular}{|c|c|c|c|c|} \hline
$n$   & \multicolumn{4}{|c|}{$4m$} \\ \hline
 $m$   & odd & \multicolumn{3}{|c|}{even} \\ \hline
 $e$   & $e_1$  & \multicolumn{3}{|c|}{$e_1$} \\ \hline
 $f$ & $e_{m+1}$ & $e_{m}$ & $e_{m+1}$ & $e_{m+2}$ \\ \hline
 $t_1$ & $s_{\frac{m+1}{2}}$  & $s_{\frac{m}{2}}$ & $s_{\frac{m+1 \pm 1}{2}}$ & $s_{\frac{m}{2}+1}$ \\ \hline
 $t_2$ & $s_{\frac{3m+1}{2}}$ & $s_{\frac{3m}{2}}$ & $s_{\frac{3m+1 \pm 1}{2}}$ & $s_{\frac{3m}{2}+1}$ \\ \hline
 $t_3$ & $s_{\frac{5m+1}{2}}$ & $s_{\frac{5m}{2}}$ & $s_{\frac{5m+1 \pm 1}{2}}$ & $s_{\frac{5m}{2}+1}$ \\ \hline
 $t_4$ & $s_{\frac{7m+1}{2}}$ & $s_{\frac{7m}{2}}$ & $s_{\frac{7m+1 \pm 1}{2}}$ & $s_{\frac{7m}{2}+1}$ \\ \hline
 $\bar{P}_n$ & $\frac{1}{2}\left(1+ \frac{r_n^2}{\sqrt{2}} \right)$ & \multicolumn{3}{|c|}{$\frac{1}{2}\left(1+ \frac{1}{\sqrt{2}} \right)$ }  \\ \hline
\end{tabular}\vspace*{0.4cm}
\caption{\label{table1} The optimal effects and states for the random access test with two dichotomic measurements on even polygon state spaces $\state_n$, where $n=4m$ for some $m \in \nat$. When $m$ is even, the optimizing effect $f$ is not unique and thus also the optimizing states are different for different choises of $f$ and furthermore they may not be unique even for a fixed choice of $f$.}
\end{table}

Let us then consider the case when $n\neq 4m$ for any $m \in \nat$ so that we must actually then have that $n = 4m+2$ for some $m \in \nat$. In this case we can again distinguish two different cases: when $m$ is odd and when $m$ is even. To see the reason for this, let us consider the expression for $\bar{P}_n$ when $n=4m+2$. In this case we see that the maximum possible value for the inner supremums can be attained if $j_1=k/2$, $j_2=k/2+m+1/2$, $j_3=k/2+2m$ and $j_4=k/2+3m+3/2$. However, we see that these are not integers neither for odd or even $k$. 

If we assume that $k$ is even, then the closest integers with which we can attain the supremum are $j'_1=k/2$, $j'_{\pm 2}=(k\pm 1)/2+m+1/2$, $j'_3=k/2+2m$ and $j'_{\pm 4}=(k\pm 1)/2+3m+3/2$. For these parameters we get that
\begin{align*}
\bar{P}_n &= \sup_{k \in \{1, \ldots, n/2\}} \frac{1}{8} \left[ 4 + 2 r_n^2 \cos\left(\frac{(k-1)\pi}{n} \right) + 2  \sin\left(\frac{(k-1)\pi}{n} \right) \right]
\end{align*}

It can be shown that in this case the optimizing value for $k$ can be restricted to be either $m+1$ or $m+2$. Since we are assuming that $k$ is even, this leads to two different considerations: when $m$ is even and when $m$ is odd. For odd $m$ we have that the optimal value is $k = m+1$ and in this case Eq. \eqref{eq:n=4m+2-odd} holds. Similarly for even $m$ we have that the optimal value is $k = m+2$ and in this case Eq. \eqref{eq:n=4m+2-even} holds.

On the other hand, if we assume that $k$ is odd, then the optimizing values for $j_i$'s are $j''_{\pm 1}=(k\pm 1)/2$, $j''_{2}=k/2+m+1/2$, $j''_{\pm 3}=(k\pm 1)/2+2m$ and $j''_{4}=k/2+3m+3/2$. For these parameters we get that
\begin{align*}
\bar{P}_n &= \sup_{k \in \{1, \ldots, n/2\}} \frac{1}{8} \left[ 4 + 2  \cos\left(\frac{(k-1)\pi}{n} \right) + 2  r_n^2 \sin\left(\frac{(k-1)\pi}{n} \right) \right]
\end{align*}
Also in this case the optimizing value for $k$ can be shown to be either $m+1$ or $m+2$. Again we have different cases for even and odd $m$. For odd $m$ we have that the optimal value is $k = m+2$ and then one can confirm that we get Eq. \eqref{eq:n=4m+2-odd}. Similarly for even $m$ we have that the optimal value is $k = m+1$ and then one can confirm that we get Eq. \eqref{eq:n=4m+2-even}.  

To conclude, the maximum success probabilities in the case $n=4m+2$ are given by Eq. \eqref{eq:n=4m+2-odd} when $m$ is odd and by Eq. \eqref{eq:n=4m+2-even} when $m$ is even. One can show that both of these expressions are strictly between the qubit (and the rebit) value $1/2(1+1/\sqrt{2})$ and the upper bound given in Eq. \eqref{eq:P_n-even-bound} for all finite $m \in \nat$. In the limit $m \to \infty$ the maximum success probability approaches the qubit value. All the optimal effects and states are explicitly expressed in Table \ref{table2}.

\begin{table}
\centering
\begin{tabular}{|c|c|c|c|c|} \hline
$n$   & \multicolumn{4}{|c|}{$4m+2$} \\ \hline
 $m$   & \multicolumn{2}{|c|}{odd} & \multicolumn{2}{|c|}{even} \\ \hline
 $e$   & \multicolumn{2}{|c|}{$e_1$}  & \multicolumn{2}{|c|}{$e_1$} \\ \hline
 $f$ & $e_{m+1}$ & $e_{m+2}$ & $e_{m+1}$ & $e_{m+2}$  \\ \hline
 $t_1$ & $s_{\frac{m+1}{2}}$  & $s_{\frac{m\pm 1}{2}+1}$ & $s_{\frac{m+1\pm 1}{2}}$ & $s_{\frac{m}{2}+1}$\\ \hline
 $t_2$ & $s_{\frac{3m\pm 1}{2}+1}$ & $s_{\frac{3m+1}{2}+1}$ & $s_{\frac{3m}{2}+1}$ & $s_{\frac{3m+1\pm 1}{2}+1}$\\ \hline
 $t_3$ & $s_{\frac{5m+1}{2}+1}$ & $s_{\frac{5m\pm 1}{2}+2}$ & $s_{\frac{5m+1\pm 1}{2}+1}$ & $s_{\frac{5m}{2}+2}$\\ \hline
 $t_4$ & $s_{\frac{7m \pm 1}{2}+2}$ & $s_{\frac{7m+1}{2}+2}$ & $s_{\frac{7m}{2}+2}$ & $s_{\frac{7m+1\pm 1}{2}+2}$\\ \hline 
& \multicolumn{2}{|c|}{} & \multicolumn{2}{|c|}{} \\[-0.4cm] 
$\bar{P}_n$ & \multicolumn{2}{|c|}{
 $\frac{1}{4} \left[ 2 +  r_n^2 \cos\left(\frac{m\pi}{n} \right) +    \sin\left(\frac{m\pi}{n} \right) \right]$} & \multicolumn{2}{|c|}{
 $\frac{1}{4} \left[ 2 +   \cos\left(\frac{m\pi}{n} \right) +   r_n^2 \sin\left(\frac{m\pi}{n} \right) \right]$}  \\[0.1cm] \hline
\end{tabular}\vspace*{0.4cm}
\caption{\label{table2}  The optimal effects and states for the random access test with two dichotomic measurements on even polygon state spaces $\state_n$, where $n=4m+2$ for some $m \in \nat$. For both even and odd $m$ the optimizing effect $f$ is not unique and thus also the optimizing states are different for different choices of $f$ and furthermore they may not be unique even for a fixed choice of $f$.}
\end{table}

\subsubsection{Comparing rebit to the even polygon state spaces} 
As was established in the previous section, instead of having a single oscillatory upper bound that would explain the behaviour of the maximum success probability depicted in Fig. \ref{fig:polygon-max} we must in fact divide the even polygons in four distinct cases: i) $n=4m$ and $m$ is odd, ii) $n=4m$ and $m$ is even, iii) $n=4m+2$ and $m$ is odd, and iv) $n=4m+2$ and $m$ is even for some $m \in \nat$. 

\begin{figure}[t]
\centering
\includegraphics[scale=0.13]{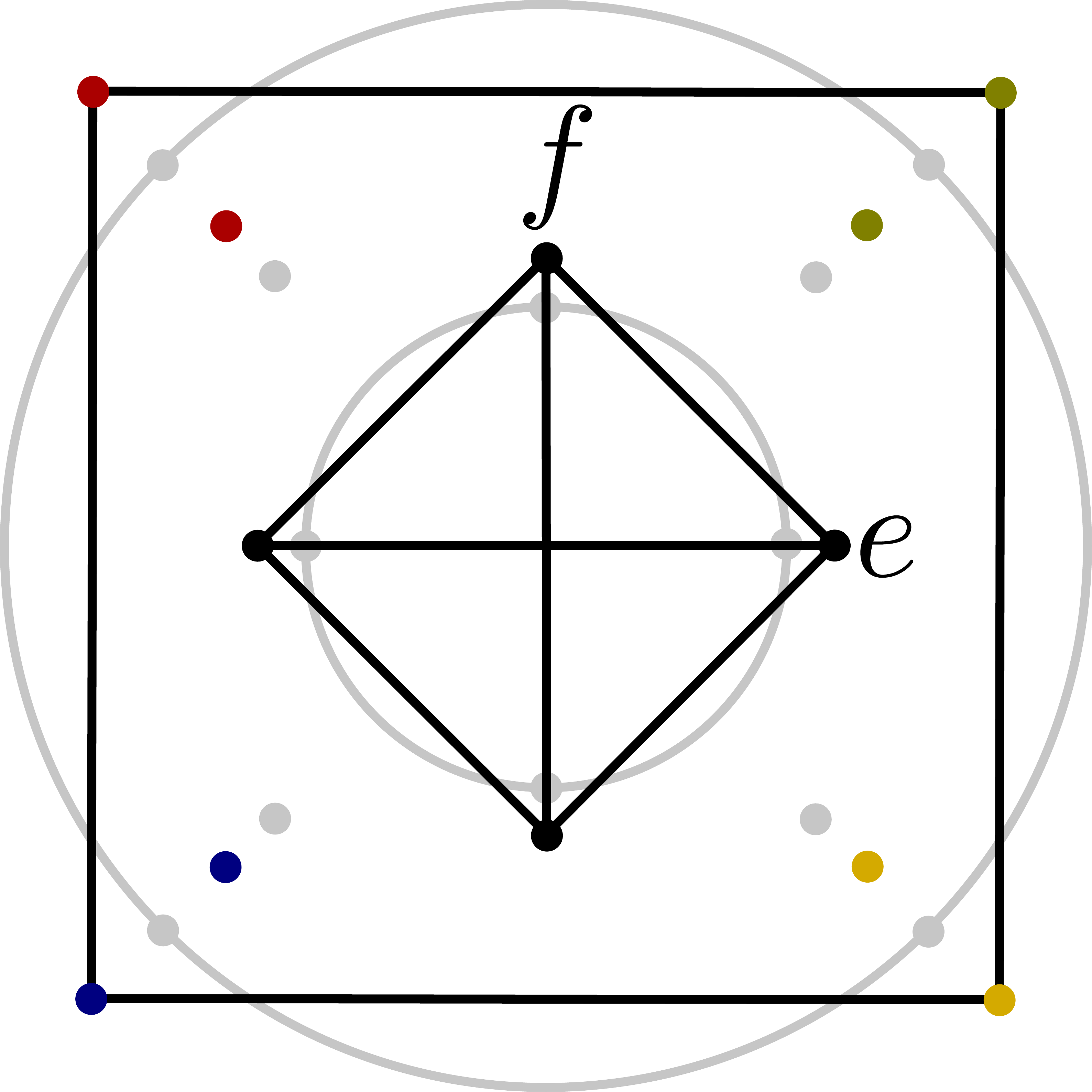} \ 
\includegraphics[scale=0.13]{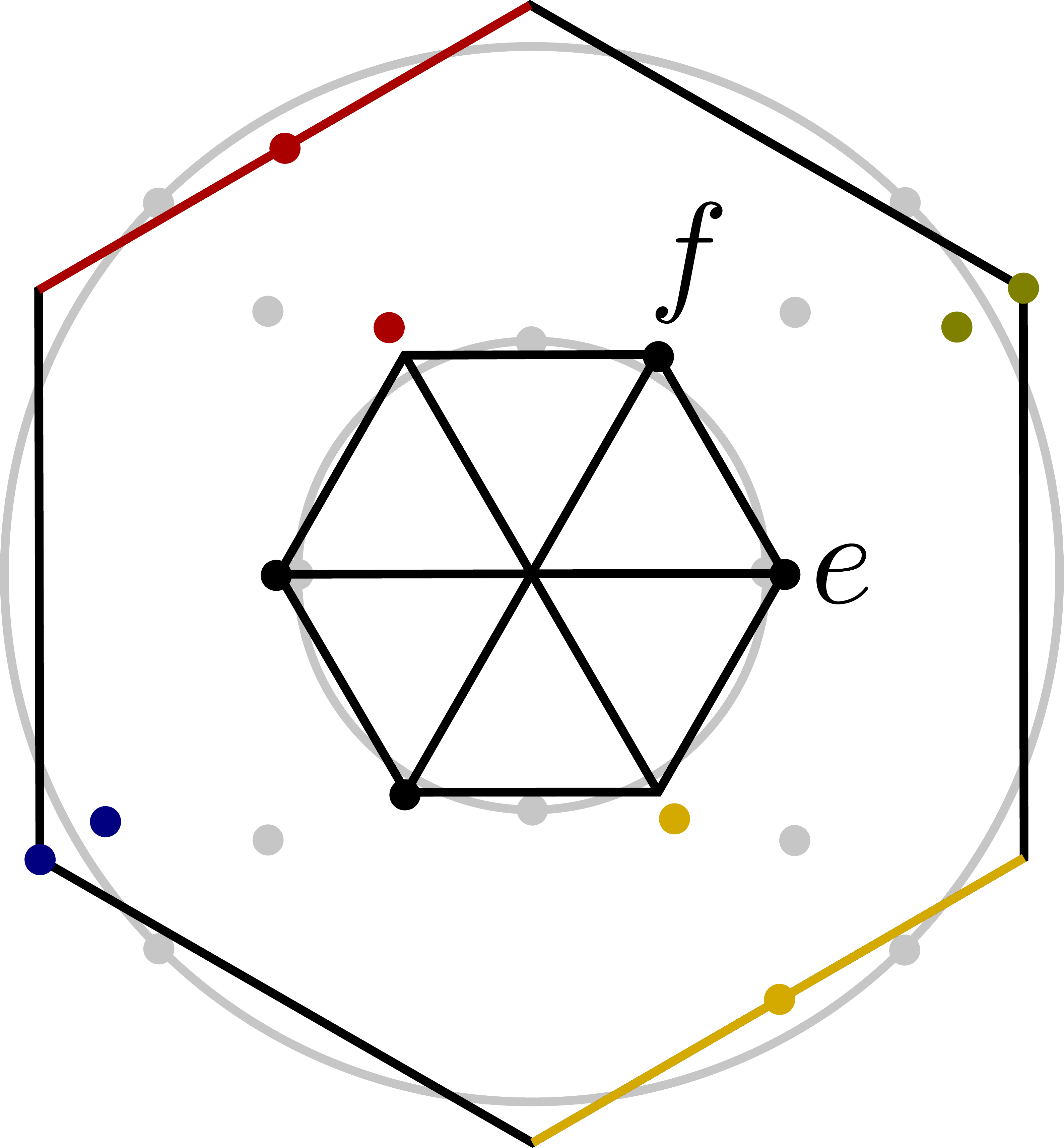} \\
\ \\
\includegraphics[scale=0.13]{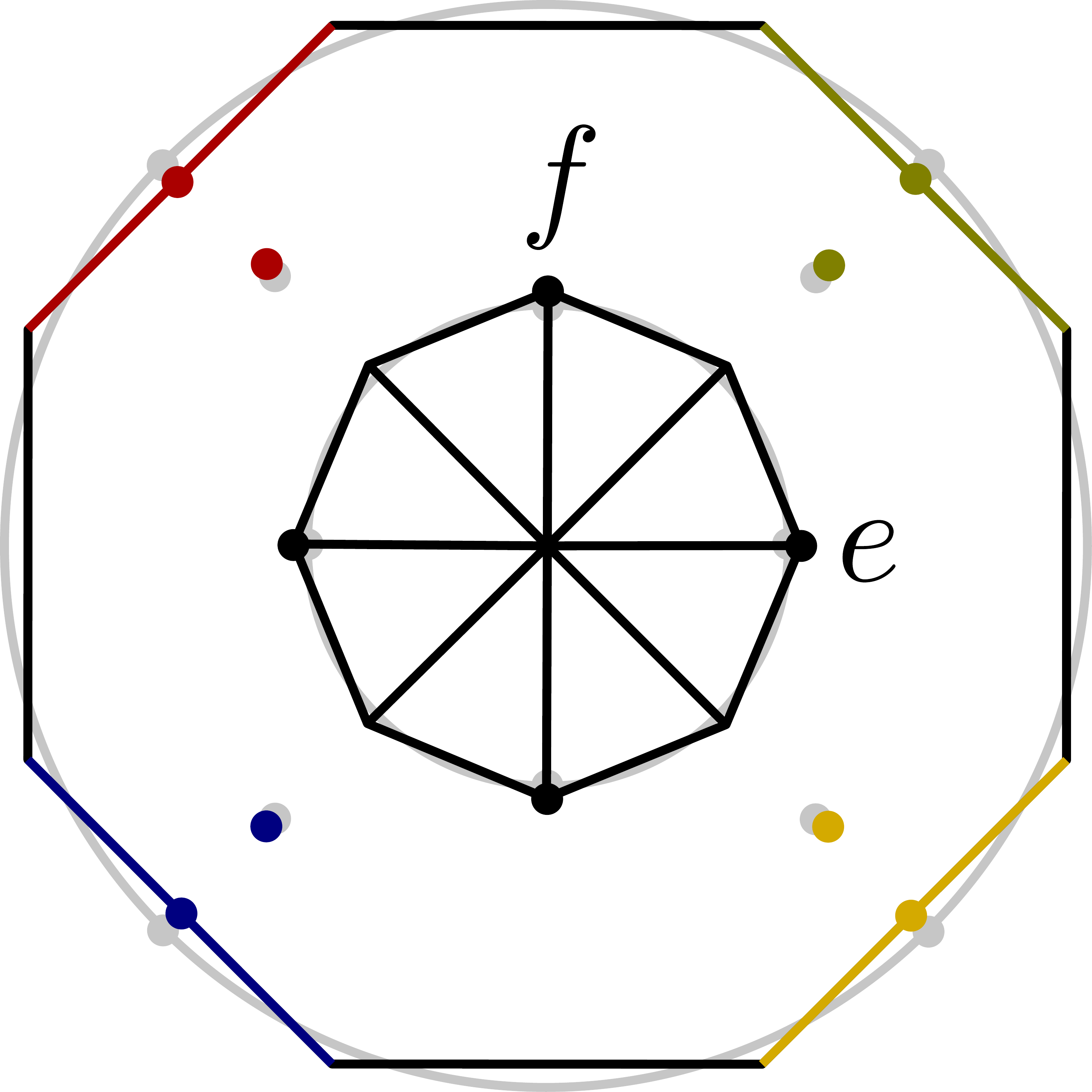} \ 
\includegraphics[scale=0.13]{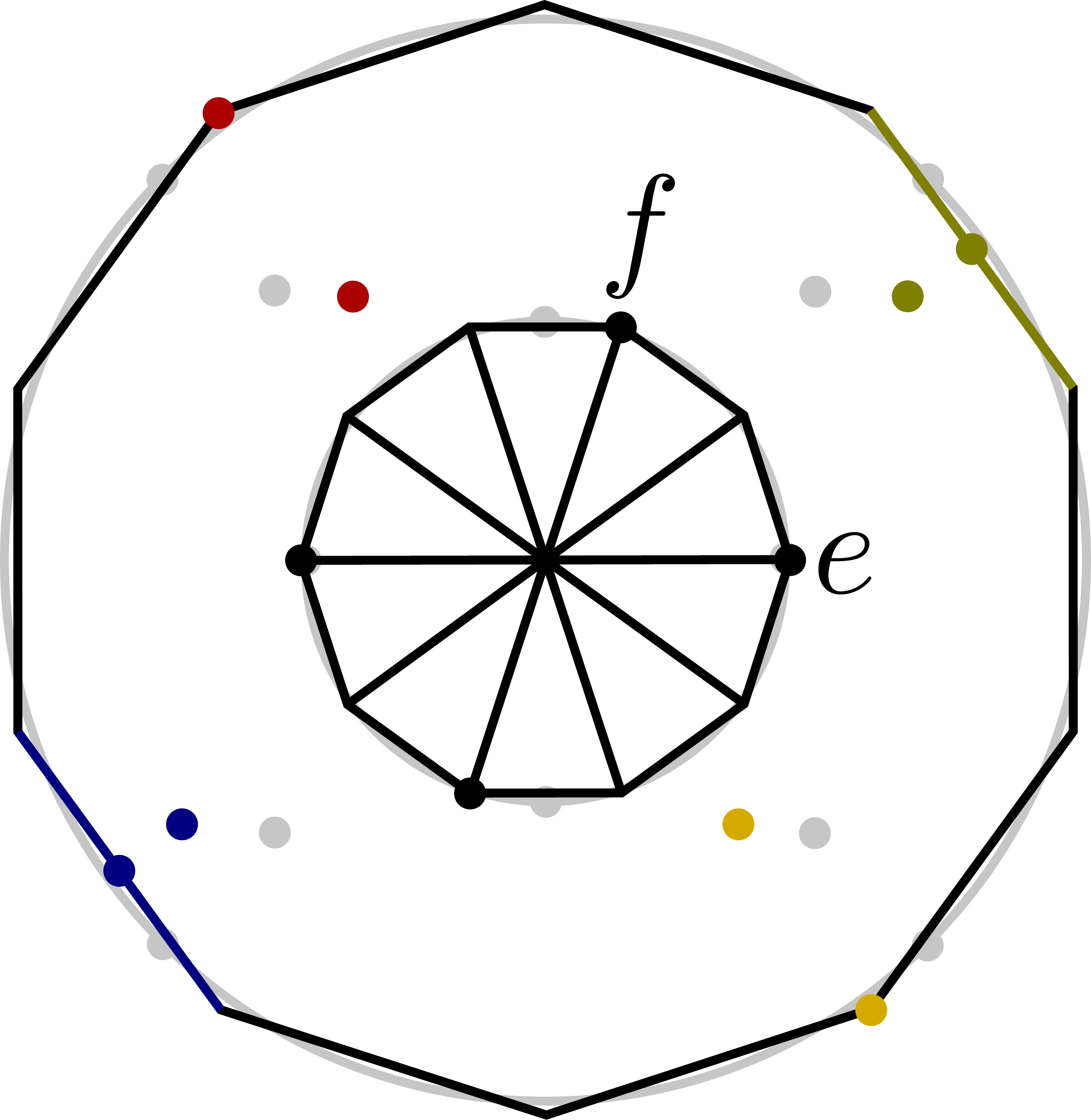} \ 
\caption{ \label{fig:optimal-4-6-8-10} Some of the maximizing effects and states for the first four even polygons. The sums of the maximizing effects and their maximizing states are represented by the same colors as in Fig. \ref{fig:rebit}. For comparison, the optimal effects and states in rebit are denoted in gray.}
\end{figure}

Regarding the geometry of the optimal effects and the states, as was explained previously, we can freely choose $e=e_1$. According to our analysis, similarly to the rebit case, the second extreme effect $f$ is aligned so that it is `furthest' away from both $e$ and $u-e$ along the circumsphere of the polygon, i.e., taking values $f=e_k$ where $k \in \{m,m+1,m+2\}$ depending on the polygon. Finally, again just as in the rebit, the optimizing states $t_1, t_2,t_3,t_4$ are (not necessarily uniquely) determined by the sums $e+f$, $u-e+f$, $e+u-f$ and $u-e+u-f$ as close to the same directions of these sums in the 2-dimensional $(x,y)$-projections as possible. 

We illustrate the previous properties along with pointing out the differences of the four aforementioned cases by depicting some of the optimal effects and states in the first four even polygons in Fig. \ref{fig:optimal-4-6-8-10}. We see the first two cases when $n=4m$ with $m=1$ and $m=2$ depicted on the left, whereas the first two cases when $n=4m+2$ with $m=1$ and $m=2$ depicted on the right. Furthermore, the top figures are examples of the cases when $m$ is odd and the bottom ones are examples of the cases when $m$ is even. The depicted states and effects are chosen from Tables \ref{table1} and \ref{table2} in such a that we take $f=e_{m+1}$ for all of the different cases and if there are more than one optimizing state for these optimizing effects, then we take the equal convex mixture of such states for the reasons explained below. 

From the figure we see that in the cases when $n=4m$ the alignment of the optimal effects $e$ and $f$ can be chosen to be just as in the rebit state space. However, the alignment of the optimal states differ for odd and even $m$: if $m$ is odd, then the optimal states $t_1, t_2, t_3, t_4$ are unique pure states and aligned in the same direction in the $(x,y)$-projection as the sums of the optimal effects (just as in the rebit), but in the case that $m$ is even, the optimal states are not unique and they can be chosen from the corresponding 2-dimensional faces of the state space. In particular, for illustration purposes, for even $m$ we chose the optimal states to be of the form $t_i=1/2(s_{j_i}+s_{j_i+1})$ for some $j_i \in \{1, \ldots, n\}$ for all $i \in \{1,2,3,4\}$ so that they align in the same direction as the sums of the optimal effects $e$ and $f$ in the 2-dimensional $(x,y)$-projection just as in the case of odd $m$ and the rebit. However, in the case that $n=4m+2$ already for $m=1$ and $m=2$ the situation differs from that of the rebit and the qubit since the effect $f$ cannot be chosen to be orthogonal to $e$ in the $(x,y)$-projection in the $z=1/2$--plane as before. For this reason also the optimal states are determined differently: after fixing the effects $e=e_1$ and $f=e_{m+1}$, for odd $m$ we have that $t_1$ and $t_3$ are uniquely determined and $t_2$ and $t_4$ are not (and we choose them along the same direction with the sums of the respective optimizing effects), whereas for even $m$ it is the opposite.

\subsubsection{Odd polygons}

In odd polygons theories the maximum success probability $\bar{P}_n$ for two dichotomic measurements depends on $n$ as follows:
\begin{align}
& \bar{P}_n = \frac{1}{4} \left[ 2 +   \cos\left(\tfrac{m\pi}{n} \right) +   r_{2n}^2 \sin\left(\tfrac{m\pi}{n} \right) \right] \quad \textrm{ $n=4m+1$ for $m \in \nat$}\label{eq:n=4m+1} \\
& \bar{P}_n = \frac{1}{4} \left[ 2 +   \cos\left(\tfrac{(m+1)\pi}{n} \right) +   r_{2n}^2 \sin\left(\tfrac{(m+1)\pi}{n} \right) \right]  \  \textrm{ $n=4m+3$ for $m \in \nat$},\label{eq:n=4m+3}
\end{align}
where $r_{2n} = \sqrt{\sec\left( \frac{\pi}{2n}\right)}$.

The proof of these results is similar to our previous proof in case of the even polygon theories. In Eq. \eqref{eq:P-2-2} we can fix the first extreme effect $e$ to be any of the nontrivial extreme effects and for odd polygons we choose $e=g_1$. Furthermore, we have that we can choose $f = g_k$ for some $k \in \{1, \ldots, n\}$ and $t_i = s_{j_i}$ for some $j_i \in \{1, \ldots, n\}$ for all $i \in \{1,2,3,4\}$. We can write $\bar{P}_n$ for odd polygons as
\begin{align*}
\bar{P}_n&= \sup_{k \in \{1, \ldots, n\}} \frac{1}{8} \left[ \sup_{j_1 \in \{1, \ldots, n\}} (g_1+g_k)(s_{j_1}) + \sup_{j_2 \in \{1, \ldots, n\}} (u-g_1+g_k)(s_{j_2}) \right. \\
&\ \ \ \left.   + \sup_{j_3 \in \{1, \ldots, n\}} (u-g_1+u-g_k)(s_{j_3}) + \sup_{j_4 \in \{1, \ldots, n\}} (g_1+u-g_k)(s_{j_4})\right].
\end{align*}
By expanding the previous expression and by using some trigonometric identities we can rewrite the previous equation as
\begin{align*}
\bar{P}_n&= \sup_{k \in \{1, \ldots, n\}} \frac{1}{8} \left[\sup_{j_1 \in \{1, \ldots, n\}} \frac{1}{1+r_n^2}\left(2+ 2 r_n^2 \cos\left(\tfrac{(k-1)\pi}{n} \right) \cos\left(\tfrac{(k+1-2j_1)\pi}{n} \right) \right)  \right. \\
& \quad \quad  \ \ + \sup_{j_2 \in \{1, \ldots, n\}} \frac{1}{1+r_n^2}\left(1+r_n^2- 2r_n^2 \sin\left(\tfrac{(k-1)\pi}{n} \right) \sin\left(\tfrac{(k+1-2j_2)\pi}{n} \right) \right) \\
& \quad \quad  \ \ + \sup_{j_3 \in \{1, \ldots, n\}} \frac{1}{1+r_n^2}\left(2 r_n^2- 2 r_n^2 \cos\left(\tfrac{(k-1)\pi}{n} \right) \cos\left(\tfrac{(k+1-2j_3)\pi}{n} \right) \right) \\
 &  \left.  \quad \quad   \ \ + \sup_{j_4 \in \{1, \ldots, n\}} \frac{1}{1+r_n^2}\left(1+r_n^2+ 2r_n^2 \sin\left(\tfrac{(k-1)\pi}{n} \right) \sin\left(\tfrac{(k+1-2j_4)\pi}{n} \right) \right) \right]. 
\end{align*}
Similarly to as we did in the even polygons we can confirm that the algebraic maximum, and thus and upper bound for the actual maximum success probability, can be attained with parameters $k=1+n/4$, $j_1=(k+1)/2$, $j_2=(k+1)/2 + n/4$, $j_3 = (k+1)/2+n/2$ and $j_4 = (k+1)/2+3n/4$ in which case we have that
\begin{align*}
\bar{P}_n \leq \frac{1}{2} \left( 1+  \frac{ \sqrt{2}r_n^2}{1+r_n^2} \right) = \frac{1}{2} \left( 1+  \frac{ \sqrt{2}}{1+\cos\left( \frac{\pi}{n} \right)} \right).
\end{align*}
We note that for any finite $n$ this bound is always larger than the maximum success probability in qubit and in the limit $n \to \infty$ the bound coincides with the qubit value. However, unlike in the even polygon case, we cannot attain this upper bound in any odd polygon since $k=1+n/4$ will never be an integer for odd $n$. But since the optimized expressions are of the similar form as for even polygons we again know that the actual supremum values will be attained for parameter values close to those that attain the algebraic maximum. All the optimizing effects and states are presented in Table \ref{table3}. We depict some of these optimal effects and states in Fig. \ref{fig:optimal-5-7-9-11} in similar fashion as we did in Fig. \ref{fig:optimal-4-6-8-10} for even polygons.

\begin{figure}[H]
\centering
\includegraphics[scale=0.13]{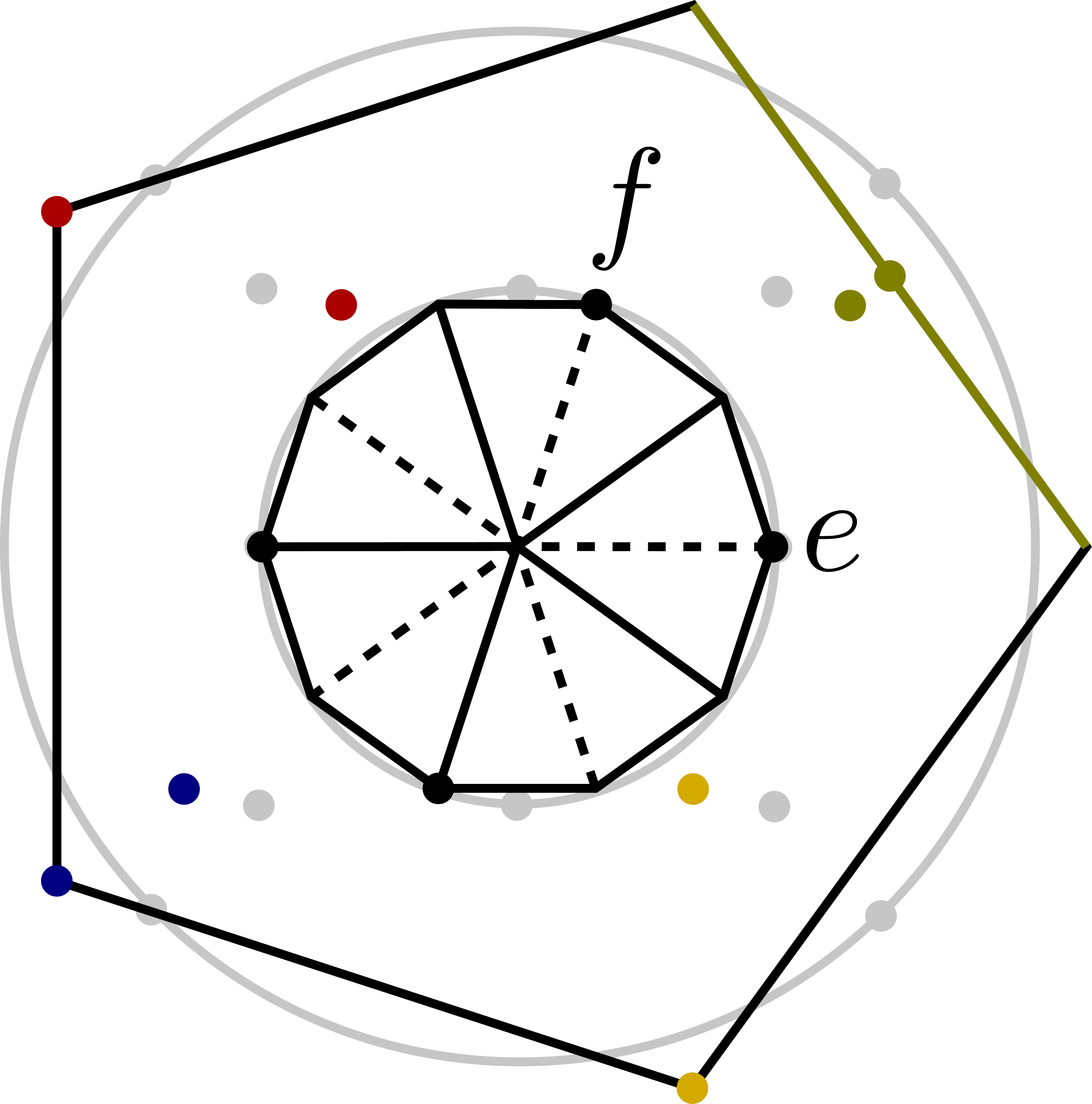} \ 
\includegraphics[scale=0.13]{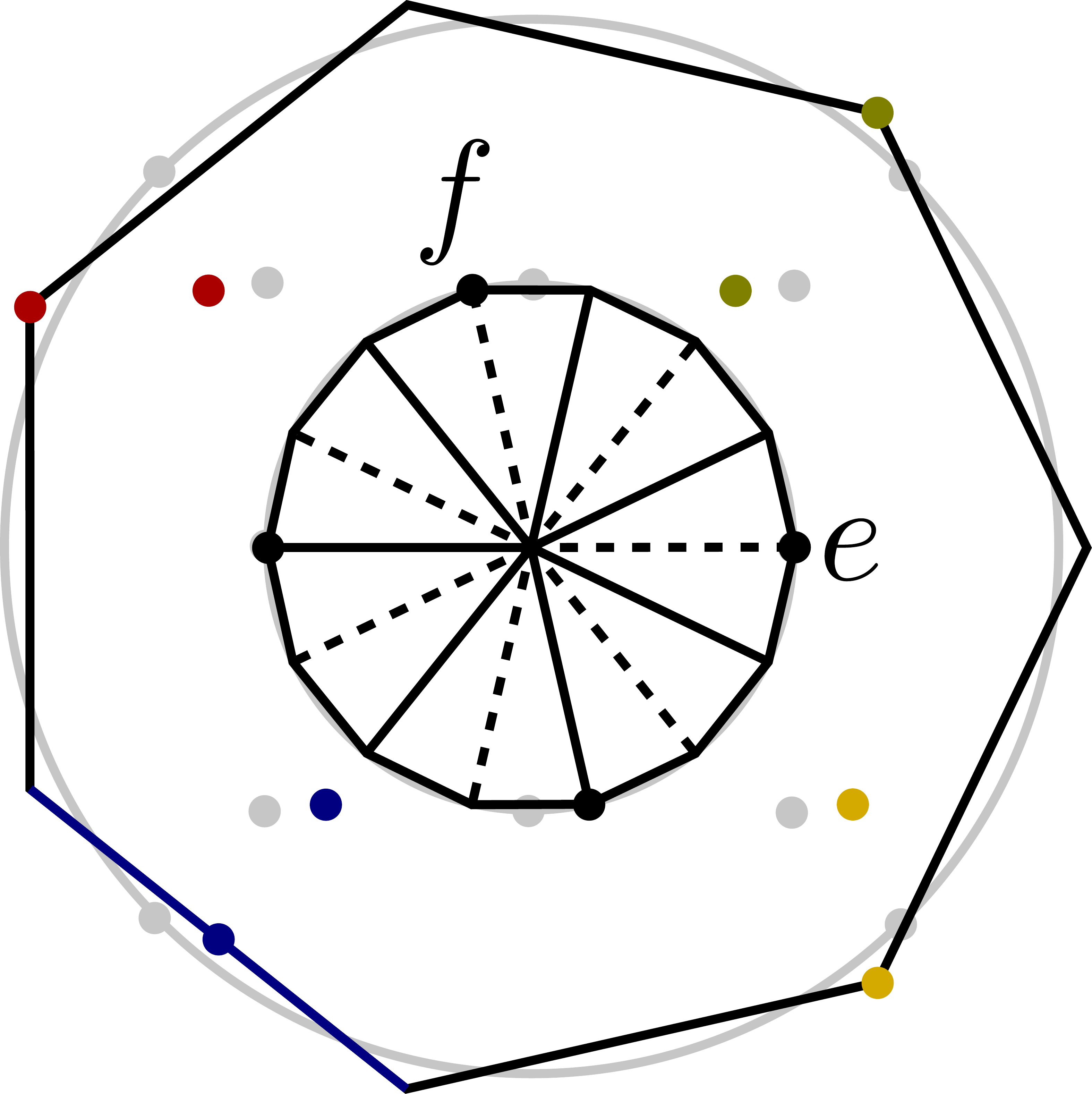} \\
\ \\
\includegraphics[scale=0.13]{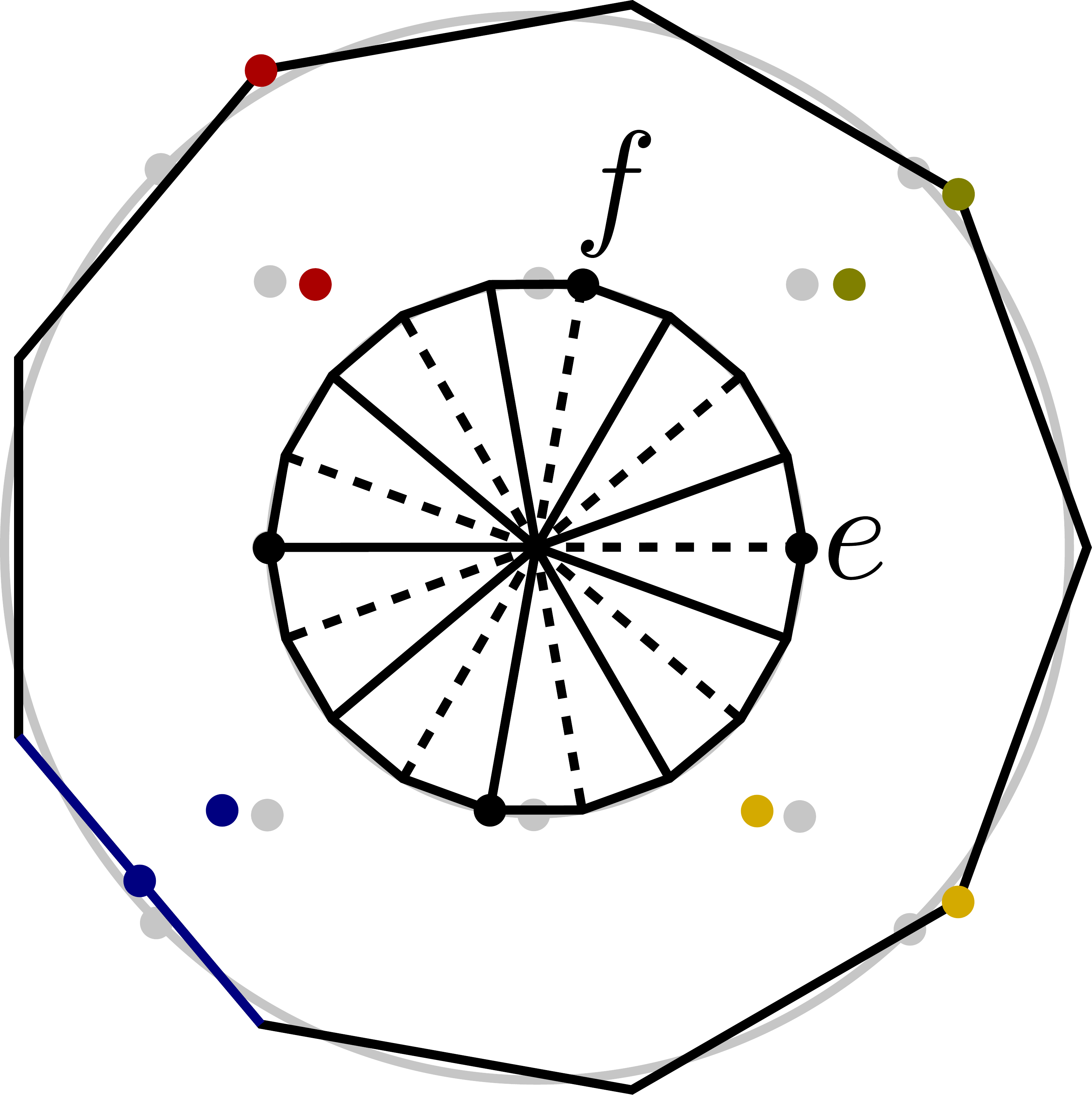} \ 
\includegraphics[scale=0.13]{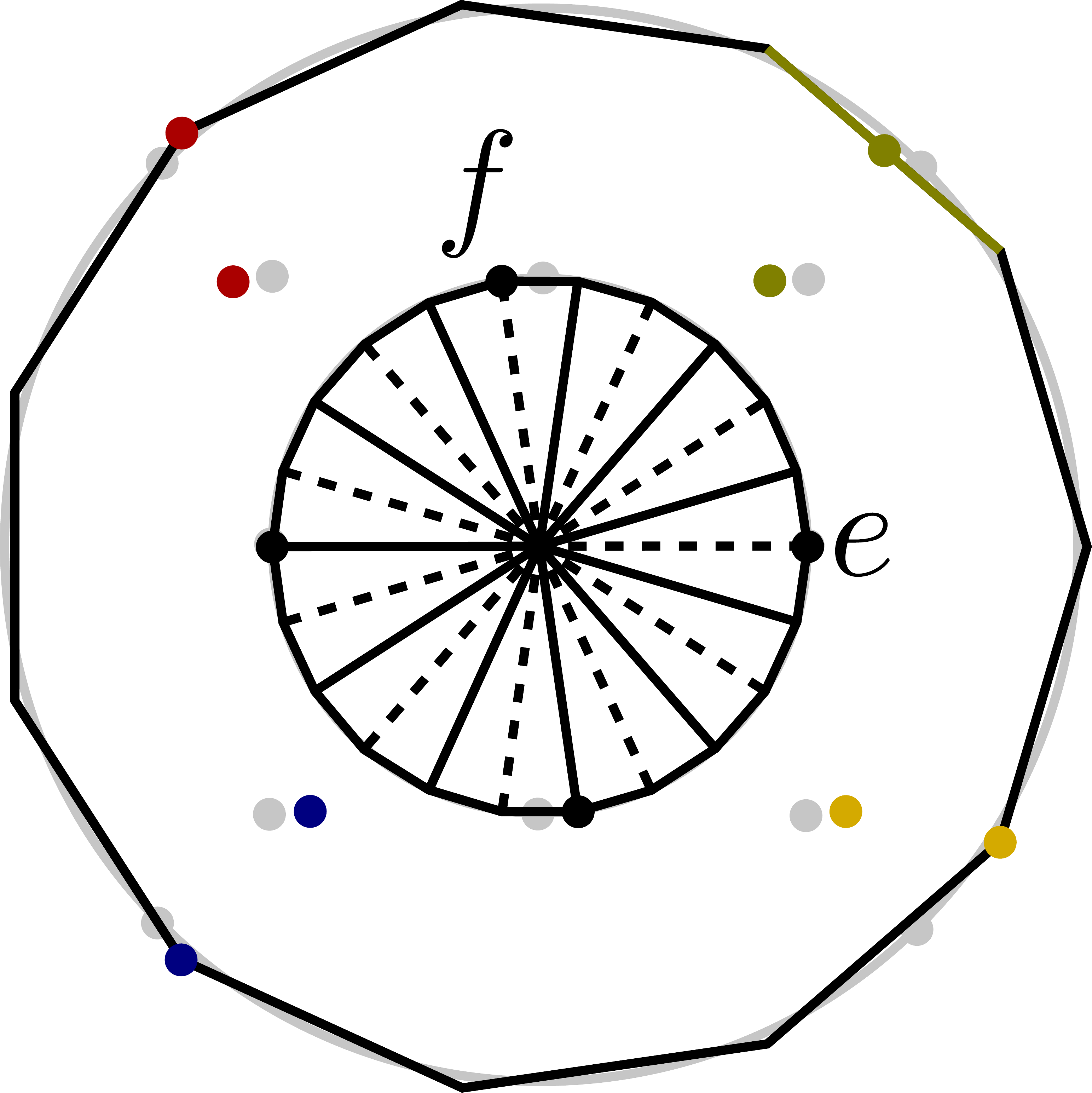} \ 
\caption{ \label{fig:optimal-5-7-9-11} Some of the maximizing effects and states for the first four odd polygons. The sums of the maximizing effects and their maximizing states are represented by the same colors as in Fig. \ref{fig:rebit}. For comparison, the optimal effects and states in rebit are denoted in gray.}
\end{figure}

\begin{table}[H]
\centering
\begin{tabular}{|c|c|c|c|c|} \hline
 $n$   & \multicolumn{2}{|c|}{$4m+1$} & \multicolumn{2}{|c|}{$4m+3$} \\ \hline
  $m$ & odd & even & odd & even  \\ \hline
 $e$   & $g_1$  & $g_1$ & $g_1$ & $g_1$ \\ \hline
 $f$ & $g_{m+1}$ & $g_{m+1}$ & $g_{m+2}$ & $g_{m+2}$  \\ \hline
 $t_1$   & $s_{\frac{m\pm 1}{2}+1}$ & $s_{\frac{m}{2}+1}$ & $s_{\frac{m+1}{2}+1}$ & $s_{\frac{m+1\pm1}{2}+1}$\\ \hline
 $t_2$  & $s_{\frac{3m+1}{2}+1}$ & $s_{\frac{3m}{2}+1}$ & $s_{\frac{3m+1}{2}+2}$ & $s_{\frac{3m}{2}+2}$\\ \hline
 $t_3$  & $s_{\frac{5m+ 1}{2}+1}$ & $s_{\frac{5m+1\pm 1}{2}+1}$ & $s_{\frac{5m\pm 1}{2}+3}$ & $s_{\frac{5m}{2}+3}$\\ \hline
 $t_4$  & $s_{\frac{7m+1}{2}+1}$ & $s_{\frac{7m}{2}+2}$ & $s_{\frac{7m+1}{2}+3}$ & $s_{\frac{7m}{2}+4}$\\ \hline
 & \multicolumn{2}{|c|}{} & \multicolumn{2}{|c|}{} \\[-0.4cm] 
 $\bar{P}_n$ & \multicolumn{2}{|c|}{
 $\frac{1}{4} \left[ 2 +   \cos\left(\frac{m\pi}{n} \right) +   r_{2n}^2 \sin\left(\frac{m\pi}{n} \right) \right]$} & \multicolumn{2}{|c|}{
 $\frac{1}{4} \left[ 2 +   \cos\left(\frac{(m+1)\pi}{n} \right) +   r_{2n}^2 \sin\left(\frac{(m+1)\pi}{n} \right) \right]$}  \\[0.1cm] \hline
\end{tabular}
\caption{\label{table3} The optimal effects and states for the random access test with two dichotomic measurements on odd polygon state spaces. For both cases $n=4m+1$ and $n=4m+3$ the optimizing states are different for odd and even $m$ and they may not be unique. }
\end{table}

\section{Conclusions}

We introduced a generalization of the well-known quantum random access code (QRAC) information processing tasks, namely the random access tests (RATs), in the framework of general probabilistic theories. In particular, we formulated the random access tests to study the properties of the measurements that are used to decode the information in the test. We showed that the figure of merit of these tasks, the average success probability, is linked to the decoding power of the harmonic approximate joint measurement of the used measurements. We saw that in this way the harmonic approximate joint measurement, which can be defined for any set of measurements, can be used to give upper bounds for the maximum average success probability of the RAT of the given measurements.

In quantum theory it was previously known that one has to use incompatible measurements in order to obtain a quantum advantage in QRACs over the classical case. We generalized this result to show that in order to obtain an advantage in RATs compared to the classical case either the measurements have to be incompatible or the theory itself must possess a property which we call super information storability meaning that the information storability is strictly larger than the operational dimension of the theory. In the case of quantum and point-symmetric state spaces super information storability does not hold so that in these cases the violation of the classical bound implies incompatibility of the measurements. In general, our result can be used as a semi-device independent certification of incompatibility in GPTs. We showed that maximal incompatibility of two dichotomic measurements, i.e., maximal robustness under noise in the form of mixing, links to their performance in the RAT. More precisely, we proved that two dichotomic measurements are maximally incompatible if and only if they can be used to accomplish the RAT with certainty by using a set of affinely dependent states. 

As examples of state spaces other than quantum and classical we considered the regular polygon theories in which we exhaustively examined the RATs with two dichotomic measurements (as their operational dimension is also two). The even polygons are point-symmetric so that in this case they do not have super information storability but for odd polygons we gave an explicit construction of compatible measurements that violate the classical bound, hence detecting the super information storability property of these theories. Furthermore, we solved the optimal success probabilities of these RATs in all polygons and saw that for incompatible measurements it is possible to violate the quantum bound as well, both in even and odd polygons.


\begin{thebibliography}{10}
\bibitem{Plavala16}
M.~Pl\'avala.
\newblock All measurements in a probabilistic theory are compatible if and only
  if the state space is a simplex.
\newblock {\em Phys. Rev. A}, 94:042108, 2016.
\end{thebibliography}
\end{document}